\documentclass[11pt]{article}

\usepackage[utf8]{inputenc}
\usepackage{
    amsmath,
    amsfonts, 
    amsthm, 
    amssymb,
    authblk,
    booktabs,
    dcolumn,
    enumerate,
    float,
    fullpage,
    hyperref,
    mathrsfs,
    multirow,
    siunitx,
    tabularx,
    titlesec,
    tikz,
}
\usepackage[numbers]{natbib}

\renewcommand{\emptyset}{\varnothing}

\newtheorem{theorem}{Theorem}[]
\newtheorem{proposition}{Proposition}[]
\newtheorem{lemma}{Lemma}[]
\newtheorem{corollary}{Corollary}[]

\newtheorem{observation}{Observation}[]

\newtheorem{definition}{Definition}

\newcounter{daggerfootnote}

\newcolumntype{d}[1]{D{.}{.}{#1}}

\DeclareMathVersion{nxbold}
\SetSymbolFont{operators}{nxbold}{OT1}{cmr} {b}{n}
\SetSymbolFont{letters}  {nxbold}{OML}{cmm} {b}{it}
\SetSymbolFont{symbols}  {nxbold}{OMS}{cmsy}{b}{n}

\title{Optimal Sensor Placement in Power Grids: Power Domination, Set Covering, and the Neighborhoods of Zero Forcing Forts}
\author{{\Large Logan A. Smith and Illya V. Hicks}}
\date{\today}
\affil{
    \textit{Department of Computational and Applied Mathematics}\\
    \textit{Rice University}\\
    \textit{Houston, TX 77005}}

\begin{document}

\maketitle

\begin{abstract} \noindent
  To monitor electrical activity throughout the power grid and mitigate outages, sensors known as phasor measurement units can installed. Due to implementation costs, it is desirable to minimize the number of sensors deployed while ensuring that the grid can be effectively monitored. This optimization problem motivates the graph theoretic power dominating set problem. In this paper, we propose a novel integer program for identifying minimum power dominating sets by formulating a set cover problem. This problem's constraints correspond to neighborhoods of zero forcing forts; we study their structural properties and show they can be separated, allowing the proposed model to be solved via row generation. The proposed and existing methods are compared in several computational experiments in which the proposed method consistently exhibits an order of magnitude improvement in runtime performance. \newline 
  
\end{abstract}

\section{Introduction}

In order to mitigate power outages, electric companies must monitor the power grid for electrical imbalances. To observe the current and voltage levels across the power grid, sensors known as Phasor Measurement Units (PMUs) can be installed at power substations. Although outfitting all power substations with PMUs could provide the measurements necessary to monitor the power grid, installation costs make it desirable to minimize the number of sensors implemented, while still ensuring that the network can be effectively monitored. By utilizing basic circuit analysis (e.g. Kirchoff's Voltage laws and Ohm's Law) and a partial set of measurements directly obtained with PMUs, it is often possible to uniquely determine the voltage levels and currents of power grid components which are not being directly measured. The resulting optimization problem, finding minimal arrangements of sensors that can still effectively monitor the power grid, is broadly referred to as the PMU placement problem and was originally introduced by Baldwin et al. \cite{baldwin}.

Since its introduction, many strategies for the PMU placement problem have been posed, often ensuring that the resulting solutions have added resiliencies to natural disasters or malfunctioning equipment. Notably among these are fault tolerant placement strategies \cite{pmu-fault1}, strategies which can detect incorrect measurements \cite{pmu-inj2, pmu-inj1, pmu-inj3}, and multistage problems in which PMUs are installed over several time periods \cite{pmu-multi1, pmu-multi2}. 

The PMU placement problem has also been posed in terms of graph networks, where it is known as the power dominating set problem. An electrical network can be represented as a graph $G=(V,E)$, with power substations and transmission lines corresponding to vertices and edges. The set of power substations housing PMUs then corresponds to a vertex set $S \subseteq V$. A set of rules called the power domination color changing rules can be applied to $S$ in which additional vertices can also be colored. Haynes et al. \cite{haynes} introduced the first set of power dominating color changing rules and showed that if $V$ can be entirely colored by their rules, and PMUs are placed at each substation in $S$, then the state of each transmission line and power substation can be uniquely determined. Benson et al. \cite{benson} later introduced a simpler color changeing rule, and showed it to be equivalent to the color changing rule introduced by Haynes et al. \newline

\begin{minipage}{\textwidth}
\textit{Power Domination Color Changing Rule}:~\cite{benson}
\begin{enumerate}
\item[1.] Color all vertices in or adjacent to an initial vertex set $S$.
\item[2.] While there are colored vertices each adjacent to exactly one uncolored vertex, color those vertices. 
\end{enumerate}
\end{minipage} \newline \newline

The set of vertices that are colored at the end of this process is called the {\em power dominating closure} of $S$, and is denoted as $cl_P(S)$. If $cl_P(S) = V$, then $S$ is called a {\em power dominating set} in $G$. Furthermore, the cardinality of the minimum power dominating set(s) in $G$ is called the {\em power dominating number} of $G$, and is denoted as $\gamma_P(G)$. Finding minimum cardinality power dominating sets is equivalent to devising PMU placement strategies that can fully observe the network while utilizing of a minimum number of sensors.

During the first step of the power domination color changing rules, called the {\em domination step}, all uncolored vertices adjacent to an initially colored vertex are colored. In the second step, called the {\em propagation step}, additional vertices are colored if they are the unique uncolored neighbor of a colored vertex. An uncolored vertex $u$ that is colored by a vertex $v$ is said to be {\em colored} or {\em forced} by that vertex. This terminology is used due to the similarity between power domination and the related zero forcing color changing process in graphs \cite{AIM-zf}, which we briefly introduce in a later section. The propagation step corresponds to the inferences which can be made to uniquely the unmeasured states of power substations and transmission lines. This process introduces additional complexity in determining which parts of the electrical network can be observed. In many computational studies of the PMU placement problem, this step is not considered. The resulting problem more closely resembles the dominating set problem and thus computational methods for the PMU placement problem cannot be directly compared to those in power domination.

Several variations of power domination have also been proposed. These include versions in which the initial colored set must induce a connected graph \cite{conn-pd, fan&watson} or contain certain vertices \cite{bozeman, Restricted}. The computational complexity of power domination has also been a topic of interest. While the power dominating set decision problem is known to be NP-Complete even in the case of planar bipartite graphs \cite{haynes}, the power domination numbers of graphs in several classes can be computed with polynomial time complexity algorithms, including: trees~\cite{haynes}, block graphs~\cite{comp-blocks}, grid graphs~\cite{comp-grids}, interval graphs~\cite{comp-intervals}, and circular-arc graphs~\cite{comp-arcs}. Despite wide interest in power domination, relatively little work on computational methods has been published.

Currently, several integer programs (IPs) have been proposed. Each of these models have variables corresponding to the a candidate initial set of colored vertices, and the sequence in which forces occur so that all vertices are colored. Additionally, their constraints ensure that the order in which forces occur in obeys the power domination color change rules. Aazami~\cite{aazami} first introduced an IP of this style which Brimkov et al. \cite{conn-pd} improved upon by reducing the number of variables needed for the model. A third IP has also been introduced by Fan and Watson~\cite{fan&watson} which works on similar principles as the two methods previously mentioned. However, their model considers a slightly different problem than the standard power dominating set problem and cannot be directly compared.


As the main result of this paper, we propose an alternative method for computing the power domination numbers of arbitrary graphs. Our method is based on a novel integer program which exploits the neighborhood sets of structures called forts, arising in the related zero forcing problem. To formulate this IP and verify that it correctly computes the power domination numbers of arbitrary graphs, we study the structural properties of their neighborhoods, provide necessary and sufficient conditions for vertex sets to be the neighborhoods of forts, and give a practical procedure by which these neighborhoods can be identified. We also show two ways in which violated constraints for our IP can be separated, thereby allowing our model to be solved with a row generation strategy, unlike the existing methods.

In additional, we show several ways in which the runtime performance of both our method and existing methods can be improved by exploiting the presence of special optimal solutions as well as graph structures which commonly in electrical networks. We conclude with computational comparisons of the method we propose with the IP proposed by Brimkov et al. \cite{conn-pd}, and consistently demonstrate at least an order of magnitude improvement in runtime performance. 

\section{Notation}
The following notation will be used throughout this paper. A {\em graph}, $G = (V(G),E(G))$ is an ordered pair made up of a set of {\em vertices} $V(G)$, and a set of {\em edges}, $E(G) \subseteq V(G) \times V(G)$. Often these sets can be simply abbreviated as $V, E$, and the graph parameters {\em order}, $n=|V|$, and {\em size}, $m=|E|$, are used for convenience. Two vertices $u,v$ in $V$ are said to be {\em adjacent} if the edge $uv=e$ is in $E$. Additionally, $e=uv$ is said to be incident on its end vertices, $u,v \in V$. An edge $uv$ is called a {\em loop} if $u=v$ or {\em parallel} if there exist two distinct edges exist between $u,v$. All graphs in this paper are assumed to contain no loops or parallel edges. The {\em open neighborhood} of a vertex $v$ is defined as $N(v) := \{u \in V : uv \in E(G)\}$, and the {\em closed neighborhood} of $v$ is defined as $N[v] := N(v) \bigcup \{v\}$. Additionally, we define the open neighborhood of a set $S$ as $N(s) := (\cup_{v \in S} N(v))\backslash S$ and the closed neighborhood of $S$ as $N[S] = \cup_{v \in S} N[v]$. The degree of a vertex $\deg(v) := |N(v)|$.

A graph $H = (V_H, E_H)$ is said to be a {\em subgraph} of $G$ if $V_H \subseteq V$ and $E_H \subseteq E$. A subgraph $H$ can be {\em induced} in $G$ with a vertex set $S$, denoted $H = G[S] := (S, \{uv \in E : \{u,v\} \subseteq E\})$. A {\em path} in $G$ is a sequence of vertices $(v_1, v_2, \dots, v_k)$ such that for $i = 1, 2, \dots, k-1, v_i$ is adjacent to $v_{i+1}$. $G$ is called {\em connected} if for any two vertices $u,v$, there is a path from $u$ to $v$. A graph is called a {\em cycle} if it is connected and all vertices have degree $2$. A graph is called a {\em tree} if it is connected and does not contain any cycles. A {\em component} of a graph is a maximal set of vertices that induce a connected subgraph. Additional notation and well known graph properties are provided by Bondy and Murty~\cite{bondy}.


\section{Zero Forcing Fort Neighborhoods} \label{sec:structural}

The focus of this section is the structural properties of fort neighborhoods, and in particular the development of necessary and sufficient conditions for vertex sets to be fort neighborhoods. We include a brief introduction of zero forcing, derive a set cover formulation for the power dominating set problem defined with fort neighborhoods, and introduce a vertex partitioning scheme used to analyze fort neighborhoods and power dominating sets. In later sections, we show that several of these results can be exploited for computational purposes.

\subsection{Power Domination Set Cover Formulation}
Similar to power domination, zero forcing is a dynamic graph process in which an initial set of vertices are colored and additional vertices are colored as determined by a color changing rule. The zero forcing color changing rules arose in two separate fields, both with physicists studying control in quantum systems and with mathematicians seeking to bound the minimum rank of combinatorial matrices. The {\em zero forcing color changing rule}, considered for a graph $G$ and an initial set of colored vertices $S \subseteq V$, is stated below along with the definition of zero forcing forts, introduced by Brimkov et al \cite{brimkov_fast}. \newline

\begin{minipage}{\textwidth}
\textit{Zero Forcing Color Changing Rule}:~\cite{quantum-zf, AIM-zf}
\begin{enumerate}
\item[1.] While there are colored vertices each adjacent to exactly one uncolored vertex, color those vertices. 
\end{enumerate}
\end{minipage} \medskip

\begin{definition}
A non-empty set of vertices $F \subseteq V(G)$ is a {\em zero forcing fort} if for each vertex $v \in N(F)$, $|N(v) \cap F| \geq 2$.
\end{definition}

We refer to zero forcing forts simply as {\em forts}, and denote the {\em set of forts in $G$} as $\mathscr{F}(G)$. Similarly, we will refer to a set which is the closed neighborhood of a fort a {\em fort neighborhood}, and denote the {\em set of fort neighborhoods in $G$} as $\mathscr{M}(G)$. The following theorem is one of the main results posed by Brimkov et al., and characterizes the relation between forts and zero forcing sets. 

\begin{theorem} \text{\cite{brimkov_fast}} \label{thm:forts-sc}
Let $S$ be a set of vertices in a graph $G$ and $F \in \mathscr{F}(G)$. If $S$ is a zero forcing set in $G$, then $F \cap S \neq \emptyset$. 
\end{theorem}

Proposition \ref{prop:fn-setcover} is the power domination analogue of this theorem, which we show holds in both directions.

\begin{proposition} \label{prop:fn-setcover}
Let $S$ be a set of vertices in a graph $G$. Then $S$ is a power dominating set in $G$ if and only if for all $M$ in $\mathscr{M}(G), S \cap M \neq \emptyset$.
\end{proposition}
\begin{proof}
  We first prove the backwards direction by verifying its contrapositive statement. Let $G$ be a graph and let $S \subset V$ be a set that is not power dominating in $G$. Let $F := V \backslash cl_P(S) \neq \emptyset$. If any $v \in N(F)$ has exactly one $u$ neighbor in $F$, then $v$ could color $u$. Since this cannot be the case, no such $u$ can exist, thus by definition $F \in \mathscr{F}(G)$. If any vertex $v$ in $F$ is adjacent to a vertex in $S$, $v$ would be colored in the domination step, thus $v$ would be in $cl_P(S)$. Since $F \cap N[S] = \emptyset$, $N[F] \cap S = \emptyset$ and thus the contrapositive holds. 

  The other direction is shown directly. Let $S$ be a power dominating set in $G$. Consider an arbitrary $M \in \mathscr{M}(G)$ and $F \in \mathscr{F}(G)$ such that $M = N[F]$. Since $S$ is a power dominating set in $G$ if and only if $N[S]$ is a zero forcing set in $G$, by Theorem \ref{thm:forts-sc}, $N[S] \cap F \neq \emptyset$. Thus, $S \cap M = S \cap N[F] \neq \emptyset$. 
\end{proof}

Motivated by Proposition \ref{prop:fn-setcover}, we formulate Model 1 as an IP defined for an arbitrary graph $G$. In this model, each vertex in $V(G)$ is represented with a binary variable $s_v$, indicating the vertices of a set $S$. Constraint \ref{cons:1} requires that $S$ intersects each fort neighborhood present in $G$, thus a solution of Model 1 is feasible if and only if the corresponding $S$ is a power dominating set of $G$. Since the objective function of Model 1 is $|S|$, the optimum of Model 1 is $\gamma_P(G)$ and optimal solutions indicate minimum cardinality power dominating sets. \newline

\noindent \text{\textbf{Model 1: Set Cover Formulation}}
{
\begin{align}
\nonumber\min \;&\sum_{v\in V} s_v\\
\text{s.t. } \;& \sum_{v \in M} s_v \geq 1 \quad & \forall M\in \mathscr{M}(G) \label{cons:1} \\
\nonumber & s_v \in \{0,1\} & \forall v \in V
\end{align}
}
\begin{corollary} \label{cor:pd-setcover}
  For a graph $G$, $S$ is a power dominating set if and only if the characteristic vector of $S$ is a feasible solution for Model 1. Futhermore, the optimum of Model 1 is $\gamma_P(G)$.
\end{corollary}

Balas and Ng \cite{balas-ng} provide necessary and sufficient conditions for cover constraints in the form of Constraint \ref{cons:1} to be facet inducing. An immediate consequence of their work is that for any $M \in \mathscr{M}(G)$, if there exists $M \in \mathscr{M}(G)$ such that $M' \subset M$, then the cover constraint generated by $M$ is not facet inducing for the feasible region of Model 1. To evaluate Model 1 then, it suffices to consider the subset of $\mathscr{M}(G)$ consisting of minimal fort neighborhoods, denoted $\mathscr{M}^0(G) := \{M\ \in \mathscr{M}(G) : M \text{ is minimal in }\mathscr{M}(G)\}$.

Unfortunately, the number of constraints generated by minimal fort neighborhoods may still be large with respect to the orders of the graphs considered. We state the claim here but relegate its proof until later in the section, as the use of additional results are needed. 

\newcommand{\propThree}{There exists a sequence of graphs $\{G_k\}_{k=3}^\infty$ such that \\ $|\mathscr{M}^0(G_k)| = \Omega((\sqrt{n_k})!)$ as $k \to \infty$, where $n_k = |V(G_k)|$.}

\begin{proposition}\label{prop:fn-count}
\propThree
\end{proposition}

Since Model 1 may have a superpolynomial number of constraints, we adopt a constraint generation approach for solving Model 1. We note that Bozeman et al. \cite{bozeman} introduced a similar IP for power domination, and also propose the use of a constraint generation method with an auxiliary integer program. However in the constraint generation method they propose, fractional solutions of the master problem can lead their auxiliary integer program to generate constraints that are already satisfied by the fractional solution. Additionally, their constraint generation approach depends on the discovery of fort sets rather than the fort neighborhoods. The next set of observations show why fort neighborhoods are the focus of our study rather than simply the forts themselves.

\begin{figure}
	\centering
  \hfill
\begin{tikzpicture}[scale=.75]
  \begin{scope}[every node/.style={circle,draw,minimum size=2mm}]
    \def \x {0}
    \def \y {0}
    \def \spacing {2}
    \def \rspacing {1.25}
    \node[label={[rectangle]left:$v_3$}] (v3)   at (\x+.25*\spacing, \y) {};
    \node[label={[rectangle]right:$v_4$}] (v4)   at (\x+ \spacing, \y-.25*\spacing) {};
    \node[label={[rectangle]:$v_1$}] (v1)   at (\x+.25*\spacing, \y+\spacing) {};
    \node[label={[rectangle]:$v_2$}] (v2)   at (\x +\spacing, \y+1.25*\spacing) {};
    \node[label={[rectangle]:$v_5$}] (v5)   at (\x +1.625*\spacing, \y + 0.5*\spacing) {};

    \node[label={[rectangle]:}] (v8) at (\x +1.625*\spacing + 1*\rspacing, \y + 0.5*\spacing) {};
    \node (v9) at (\x +1.625*\spacing + 2*\rspacing, \y + 0.5*\spacing) {};
    \node[label={[rectangle]:$v_8$}] (v10) at (\x +1.625*\spacing + 3*\rspacing, \y + 0.5*\spacing) {};
    \node[label={[rectangle]:$v_6$}] (v6)   at (\x+1.625*\spacing+1.5*\rspacing, \y+1.25*\spacing) {};
    \node[label={[rectangle]right:$v_7$}] (v7)   at (\x+1.625*\spacing+1.5*\rspacing, \y-0.25*\spacing) {};
  \end{scope}
  \begin{scope}[every node/.style={circle,draw,minimum size=2mm}]
    \draw[very thick] (v1) to (v2)
    (v2) to (v5)
    (v5) to (v4)
    (v4) to (v3)
    (v3) to (v1)
    (v2) to (v4)

    (v5) to (v8)
    (v8) to (v9)
    (v9) to (v10)
    (v6) to (v5)
    (v6) to (v10)
    (v6) to (v8)
    (v6) to (v9)
    (v7) to (v5)
    (v7) to (v10)
    (v7) to (v8)
    (v7) to (v9);
  \end{scope}
\end{tikzpicture}
\hfill
\begin{tikzpicture}[scale=.75]
  \begin{scope}[every node/.style={circle,draw,minimum size=2mm}]
  \def \yoff {1.5}
  \def \xoff {1}
  \node[label={[rectangle]:$v_1$}] (v1) at (0, \yoff) {};
  \node[label={[rectangle]:$v_2$}] (v2) at (1*\xoff, \yoff) {};
  \node[label={[rectangle]:$v_3$}] (v3) at (2*\xoff, \yoff) {};
  \node[label={[rectangle]:$v_4$}] (v4) at (3*\xoff, \yoff) {};
  \node[label={[rectangle]:$v_5$}] (v5) at (4*\xoff, \yoff) {};
  \node[label={[rectangle]:$v_6$}] (v6) at (5*\xoff, \yoff) {};

  \node[label={[rectangle]:$v_7$}] (v7)  at (0, 0) {};
  \node[label={[rectangle]:$v_8$}] (v8)  at (1*\xoff, 0) {};
  \node[label={[rectangle]:$v_9$}] (v9)  at (2*\xoff, 0) {};
  \node[label={[rectangle]:$v_{10}$}] (v10) at (3*\xoff, 0) {};
  \node[label={[rectangle]:$v_{11}$}] (v11) at (4*\xoff, 0) {};
  \node[label={[rectangle]:$v_{12}$}] (v12) at (5*\xoff, 0) {};
  \node[label={[rectangle, xshift=.125cm]:$v_{13}$}] (v13) at (6*\xoff, 0) {};

  \node[label={[rectangle]:}] (u1) at (0, -\yoff) {};
  \node[label={[rectangle]:}] (u2) at (1*\xoff, -\yoff) {};
  \node[label={[rectangle]:}] (u3) at (2*\xoff, -\yoff) {};
  \node[label={[rectangle]:}] (u4) at (3*\xoff, -\yoff) {};
  \node[label={[rectangle]:}] (u5) at (4*\xoff, -\yoff) {};
  \node[label={[rectangle]:}] (u6) at (5*\xoff, -\yoff) {};
  \end{scope}

  \begin{scope}[every node/.style={circle,draw,minimum size=2mm}]
    \draw[very thick] (v1) to (v2)
    (v2) to (v3)
    (v3) to (v4)
    (v4) to (v5)
    (v5) to (v6)
    (v6) to (v13)
    (v7) to (v8)
    (v8) to (v9)
    (v10) to (v9)
    (v10) to (v11)
    (v11) to (v12)
    (v12) to (v13)
    (u1) to (u2)
    (u2) to (u3)
    (u3) to (u4)
    (u4) to (u5)
    (u5) to (u6)
    (u6) to (v13);
  \end{scope}
\end{tikzpicture}
\hfill
  \caption{\footnotesize In the left graph $F_1 = \{v_6, v_7\}, F_2 = \{v_1, v_2, v_4\}$ are both forts, and $|N[F_1]| > |N[F_2]|$. If Model 1 is being solved with a constraint generation approach, adding the constraint corresponding to $N[F_2]$ will reduce the number of possible solutions more sharply than adding the constraint corresponding to $N[F_1]$. In the right graph, the vertex set $F = \{v_i \in V : i \text{ is even}, 1 \leq i \leq 12\} \in \mathscr{F}(G)$, and $M = N[F] = \{v_1, \dots, v_{13}\}$. However, if any subset of $\{v_i \in V : i \text{ is odd}, 1 \leq i \leq 12\}$ is added to $F$, $F$ is still a fort with $N[F] = M$.}
    \label{fig:motivations}
\end{figure}
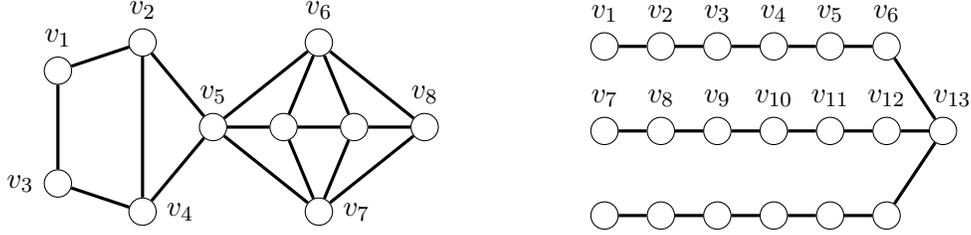
\begin{observation} \label{obs:forts-num}
 There exists a graph $G$ and $M \in \mathscr{M}(G)$ such that the bound $|\{F \in \mathscr{F}(G): N[F] = M \}| = \Omega(2^n)$ holds.
\end{observation}
\begin{observation}\label{obs:forts-neighbs}
The neighborhoods of minimum cardinality forts are not necessarily minimum cardinality fort neighborhoods, and minimum cardinality fort neighborhoods do not necessarily contain minimum weight forts. 
\end{observation}

The graph on the right in Figure \ref{fig:motivations} provides an example of a graph $G$ and fort neighborhood $M$ such that many forts have the same closed neighborhood. Observation \ref{obs:forts-neighbs} is demonstrated by the graph on the left in Figure \ref{fig:motivations}; in the example, $\{v_6, v_7\}$ is a minimum cardinality fort yet fort $\{v_1, v_2, v_4\}$ has a smaller closed neighborhood. Additionally, the second half of the observation is seen when vertices $v_1, v_2, v_4$ have weight $1$, $v_8$ has weight 10, and all other vertices have weight 0.

Generating constraints that contain a small number of variables often greatly reduces the time needed to solve set cover models like Model 1 in practice. Since these constraints are more difficult to satisfy than those with large numbers of variables, the solution space is more efficiently reduced if these stronger constraints are imposed. With Observation \ref{obs:forts-neighbs} however, we note that finding minimum cardinality forts does not necessarily lead to the discovery of minimum cardinality fort neighborhoods. Thus, if constraints are generated solely by finding forts, they are likely to be less effective than constraint generation methods that directly identify fort neighborhoods. Bearing this in mind, we turn our attention to fort neighborhoods so that such a constraint generation approach can be devised.

\subsection{Junction Vertex Partition}
We begin by introducing two sets defined for an arbitrary graph $G$: $J(G) := \{v \in V : \deg(v) \geq 3\}$, $\mathscr{P}(G) := \{V(P) : P \text{ is a component of } G[V \backslash J(G)]\}$. We call vertices in $J(G)$ {\em junctions}, vertex sets in $\mathscr{P}(G)$ {\em junction paths}, and note that $\{J(G)\} \cup P$ forms a partition of $V$. In the case that $G$ is connected and $J(G)$ is non-empty, several statements can be easily verified. 

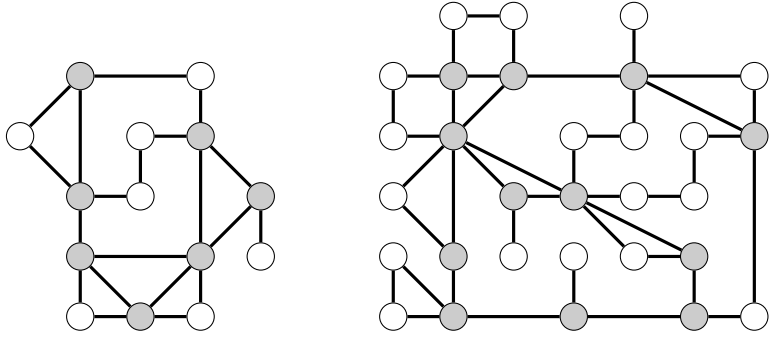
\begin{figure}
  \centering
  \begin{tikzpicture}[scale=0.4]
\begin{scope}[every node/.style={circle,draw,minimum size=2.5mm}]
	\definecolor{green}{gray}{0.8}
    \node (1) at (2,2) {};
    \node[fill=green] (2) at (4,2) {};
    \node (3) at (6,2) {};
    \node[fill=green] (4) at (6,4) {};
    \node[fill=green] (5) at (2,4) {};
    \node (6)[fill=green] at (2,6) {};
    \node[fill=green] (7) at (8,6) {};
    \node (8) at (8,4) {};
    \node (11) at (4,6) {};
    \node (10) at (4,8) {};
    \node (9)[fill=green] at  (6,8) {};
    \node (14) at (6,10) {};
    \node (13)[fill=green] at (2,10) {};
    \node (12) at (0,8) {};
    \draw[very thick] (1) to (2)
    (1) to (5)
    (2) to (3)
    (2) to (5)
    (2) to (4)
    (4) to (5)
    (3) to (4)
    (4) to (7)
    (9) to (7)
    (7) to (8)
    (4) to (9)
    (9) to (10)
    (13) to (6)
    (10) to (11)
    (6) to (11)
    (5) to (6)
    (6) to (12)
    (12) to (13)
    (14) to (9)
    (13) to (14);
\end{scope}
\end{tikzpicture} \hspace{1cm}
  \begin{tikzpicture}[scale=0.4, rotate=0]
\begin{scope}[every node/.style={circle,draw,minimum size=2.5mm}]
	\definecolor{green}{gray}{0.8}
    \node (1)  at ( 2, 12) {};
    \node (2)  at ( 4, 12) {};
    \node (3)  at ( 8, 12) {};
    \node (4)  at ( 0, 10) {};
    \node[fill=green] (5)  at ( 2, 10) {};
    \node[fill=green] (6)  at ( 4, 10) {};
    \node[fill=green] (7)  at ( 8, 10) {};
    \node (8)  at (12, 10) {};
    \node (9)  at ( 0,  8) {};
    \node[fill=green] (10) at ( 2,  8) {};
    \node (11) at ( 6,  8) {};
    \node (12) at ( 8,  8) {};
    \node (13) at (10,  8) {};
    \node[fill=green] (14) at (12,  8) {};
    \node (15) at ( 0,  6) {};
    \node[fill=green] (16) at ( 4,  6) {};
    \node[fill=green] (17) at ( 6,  6) {};
    \node (18) at ( 8,  6) {};
    \node (19) at (10,  6) {};
    \node (20) at ( 0,  4) {};
    \node[fill=green] (21) at ( 2,  4) {};
    \node (22) at ( 4,  4) {};
    \node (23) at ( 8,  4) {};
    \node[fill=green] (24) at (10,  4) {};
    \node (25) at ( 0,  2) {};
    \node[fill=green] (26) at ( 2,  2) {};
    \node[fill=green] (27) at ( 6,  2) {};
    \node[fill=green] (28) at (10,  2) {};
    \node (29) at (12,  2) {};
    \node (30) at ( 6,  4) {};
    \draw[very thick] (1) to (2)
    (1) to (5)
    (2) to (6)
    (3) to (7)
    (4) to (5)
    (4) to (9)
    (5) to (6)
    (5) to (10)
    (6) to (7)
    (6) to (10)
    (7) to (8)
    (7) to (12)
    (7) to (14)
    (8) to (14)
    (9) to (10)
    (10) to (15)
    (10) to (21)
    (10) to (16)
    (10) to (17)
    (11) to (17)
    (11) to (12)
    (13) to (14)
    (13) to (19)
    (14) to (29)
    (15) to (21)
    (16) to (17)
    (16) to (22)
    (17) to (18)
    (17) to (23)
    (17) to (24)
    (18) to (19)
    (20) to (25)
    (20) to (26)
    (21) to (26)
    (23) to (24)
    (24) to (28)
    (25) to (26)
    (26) to (27)
    (27) to (28)
    (28) to (29)
    (27) to (30);
\end{scope}
\end{tikzpicture}
  \caption{\footnotesize IEEE 14, 30 Bus Networks with $J(G)$ in gray. Proposition~\ref{prop:J-set} notes at least one minimum cardinality power dominating set is contained in $J(G)$. Shown later in Figures~\ref{tab:sep-results},~\ref{tab:full-results}, often half or fewer of the vertices in electrical network graphs are junctions.}
  \label{fig:ieee-j}
\end{figure}

\begin{proposition} \label{prop:num-jp-neighbors}
  For any $P \in \mathscr{P}(G)$, the following statements hold:
  \begin{enumerate}
  \item There exist $w_1, w_2 \in P: N(P) = (N(w_1) \cup N(w_2)) \backslash P$
  \item $1 \leq |N(P)| \leq 2$
  \item $N(P) \subseteq J(G)$
  \end{enumerate}
\end{proposition}

\begin{proof}
  To show this first statement, let $P \in \mathscr{P}(G)$. Let $p = |P|$ and $P$ be labeled $v_1, \dots, v_{p}$ such that $v_iv_{i+1} \in E$ for $i \in \{1, \dots, p-1\}$. Since $\deg(v) \leq 2$ for any $v \in P$ and $|N(v_i) \cap P| = 2$ for $i \in \{2, \dots, p-1\}$, $N(P) = (N(v_1) \backslash P) \cup (N(v_{p} \backslash P)) = (N(v_1) \cup N(v_{p})) \backslash P$.

  To show the second statement, suppose there exist vertices $u_1, u_2, u_3 \in N(P)$, distinct. Then $u_1, u_2, u_3$ are each adjacent to $v_1$ or $v_{p}$. If $p = 1$, then $\deg(v_1) \geq 3$. If not, then for at least one $j \in \{1, p\}$, $v_j$ has two neighbors in $\{u_1, u_2, u_3\}$ and one neighbor in $P$. In either case a vertex in $P$ has degree at least 3, a contradiction. Thus $|N(P)| \leq 2$. If $|N(P)| = 0$ then either $G$ is not connected or $J(G) = \emptyset$, thus $1 \leq |N(P)|$ holds.

  For the third statement, if there exist distinct $P_1, P_2 \in \mathscr{P}(G)$ and $v \in P_1, u \in P_2$ such that $vu \in E(G[V \backslash J(G)])$, then neither $P_1$ nor $P_2$ are components in $G[V \backslash J(G)]$, a contradiction. 
\end{proof}

Next we note an early observation in power domination literature, originally posed by Haynes et al.~\cite{haynes} as their Observation 4. Unfortunately, their statement did not include the necessary assumption that each component in $G$ has a degree 3 or higher vertex and is not correct as stated\footnote{The original statement does not hold for any graph that has a component with no vertices of degree 3 or more.} and a formal proof was not provided. We restate their observation in terms of $J(G)$ with an additional assumption and provide a formal proof.

\begin{proposition} \label{prop:J-set}
If $G$ is connected and $J(G)$ is non-empty, then there exists a power dominating set $S \subseteq J(G)$ such that $|S| = \gamma_P(G)$. 
\end{proposition}
\begin{proof}
  Let $S$ be a power dominating set of $G$ such that $|S| = \gamma_P(G)$. Let $S_P = S \backslash J(G)$. If $S_P = \emptyset$ then the proposition holds. Assume that $S_P \neq \emptyset$. Then there exists a vertex $v \in S_p$ and a junction path $P$ containing $v$. By Proposition \ref{prop:num-jp-neighbors}, there must exist a junction $u \in N(P)$. Let $S' = (S \cup \{u\}) \backslash \{v\}$. To show that $S'$ is a power dominating set, we will apply the power domination color change rules. Color each vertex in $N[S']$, and let $C$ denote the set of colored vertices.

  Since $v \in N(P)$, $C$ must contain some $w_1 \in P$. Since $\deg(w_1) \leq 2$ so $w_1$ can have at most uncolored neighbor, if so call it $w_2$ and add $w_2$ to $C$. Likewise, if $w_2$ exists and $w_2 \in P$, $w_2$ must have at most one uncolored neighbor $w_3$, so if it exists add $w_3$ to $C$ as well. Since $v \in P$, this process can be repeated until each vertex in $N[v]$ has been added to $C$. Since $S$ is a power dominating set in $G$, there must exist a sequence $\mathcal{S}$ in which of $V(G)$ can be colored if $S$ is the initial colored set. Let $w$ be the first vertex in $\mathcal{S}$ but not in $C$, and $w'$ be the vertex that colors it. Since $w'$ proceeds $w$ and $w$ is assumed to be the first vertex in $\mathcal{S}$ but not in $C$, $w'$ must be in $C$. Similarly, each vertex in $N(w') \backslash \{w\}$ must proceed $w$ in $\mathcal{S}$ and thus must also be in $C$. Thus, $w'$ can also color $w'$ in the constructed sequence, and $w$ can be added to $C$. This can be repeated for each subsequent vertex $w$ in $\mathscr{S}$ not in $C$, inductively, until $C = V(G)$. $S'$ is thus also a minimum cardinality power dominating set. If $S' \backslash J(G) \neq \emptyset$, this entire process can be repeated until a minimum power dominating set in $G$ is obtained comprosed entirely of vertices in $J(G)$.
\end{proof}

Since each component can be considered separately and each component without a degree three or more vertex is a path or a cycle and can be easily resolved, we can assume without loss of generality that graphs we consider going forward are connected and contain at least one vertex of degree three or more. Proposition \ref{prop:J-set} guarantees that there exist solutions of a special form for Model 1. In later sections we show that it is computationally advantageous to exploit this observation by directly finding one of these solutions.

We continue by showing that for a graph $G$, the sets in $\mathscr{P}(G)$ and fort neighborhoods in $\mathscr{M}(G)$ do not cross. Instead, for any $P \in \mathscr{G}(P), M \in \mathscr{M}(G)$, either $M \cap P = \emptyset$ or $P \subset M$. 

\begin{lemma} \label{lem:jp-contained}
Let $G$ be a graph, $M \in \mathscr{M}(G), P \in \mathscr{P}(G)$. If $M \cap P \neq \emptyset$ then $N[P] \subseteq M$.  
\end{lemma}
\begin{proof}
  Let $G$ be a graph, $M \in \mathscr{M}(G), P \in \mathscr{P}(G)$ such that $M \cap P \neq \emptyset$. Let $F \in \mathscr{M}(G)$ be a fort such that $N[F] = M$, and let $p = |P|$ and $P$ be labeled $v_1, \dots, v_{p}$ such that for all $i \in \{1, \dots, p-1\}, v_iv_{i+1} \in E$. By Proposition \ref{prop:num-jp-neighbors}, $|N(P)| \geq 1$ and since $|N(v_i) \cap P| = 2$ for all $i$ such that $2 \leq i \leq p-1$, there exists $u \in N(P)$ adjacent to $v_1$ or $v_{p}$. Assume without loss of generality that $u$ is adjacent to $v_1$.

  If $u \in F$ or $v_1 \in F$ then clearly $v_1 \in M$, so assume neither $u$ nor $v_1$ are in $F$. Since $\deg(v_1) \leq 2$, if $v_1$ has a neighbor $u_1 \neq u$ such that $u_1 \in F$, $v_1$ would have a unique neighbor in $F$, which cannot occur. Thus, $\{v_1\} \cap M = \emptyset$. In the case $p = 1$, then $P \cap M = \emptyset$. In the case that $p>1$, then $v_2 = u_2$, and thus $v_2 \not \in F$. If $v_2$ has a neighbor $u_2 \neq v_1$ such that $u_2 \in F$, $v_2$ would have a unique neighbor in $F$, so $\{v_1, v_2\} \cap M = \emptyset$. This process can be applied inductively for each vertex in $P$, showing that $\{v_1, \dots, v_{p}\} \cap M = \emptyset$. This is a contradiction, so we conclude that it must be the case that $v_1 \in F$ or $u \in F$. Moreover, since we only needed to assume that $v_1$ was adjacent to a vertex $u \not \in P$, if $v_{p}$ is adjacent to some vertex $w \not\in P$ then at least of $w, v_{p}$ must also be in $F$. Thus, $N(P) \subseteq M$.

  Let $j \in \{2, \dots, p \}$. We note that if there does not exist a $v_j, v_{j-1}$ pair such that $\{v_j, v_{j-1}\} \cap F = \emptyset$, then each vertex in $P$ is either in $F$ or has a neighbor in $F$. Thus, in this case, $P \subseteq M$. We then consider the case that there does exist a $v_j, v_{j-1}$ pair such that $\{v_j, v_{j-1}\} \cap F = \emptyset$. Then since $\deg(v_{j-1}) \leq 2$, $v_{j-2}$ (with $v_{j-2} = u$ in the case that $j=2$) cannot be in $F$ as then $v_{j-1}$ would have a unique neighbor in $F$. As before, this process can be repeated with induction to show that $v_{j-3} \dots v_{1}, u$ are not in $F$, a contradiction. Thus each $v \in P$ is either in $F$ or has a neighbor in $F$, so $P \subseteq M$. Finally, $N[P] = P \cup N(P) \subseteq M$.
\end{proof}

Since $\{J(G)\} \cup \mathscr{P}(G)$ forms a partition of $V(G)$, Lemma \ref{lem:jp-contained} implies that for any fort neighborhood $M \in \mathscr{M}(G)$, $M$ is composed of a set of junctions in $J(G)$ and a (possibly empty) set of junction paths in $\mathscr{P}(G)$. Utilizing this result, Theorem \ref{thm:NS_FN} provides a characterization of fort neighborhoods in terms of these sets. 

\begin{theorem} \label{thm:NS_FN} Let $G$ be a graph, $M \subseteq V$ be non-empty. $M \in \mathscr{M}(G)$ if and only if there exist $J_M \subseteq J(G), \mathscr{P}_M \subseteq \mathscr{P}(G)$ such that $M = J_M \cup (\cup_{P \in \mathscr{P}_M} P)$ and the following statements hold:
\begin{enumerate}
\item For every $P \in \mathscr{P}_M$ and every $v \in N(P)$, $v \in J_M$. 
\item For every $u \in J_N := \{v \in J_m : N(v) \not\subseteq M\}, |(N(u) \backslash J_N) \cap M| \geq 2$.
\end{enumerate}
\end{theorem}

\begin{proof}
  Let $G$ be a graph, $M \in \mathscr{M}(G)$. Then let $J_M := M \cap J(G)$ and $\mathscr{P}_M := \{P \in \mathscr{P}(G) : P \subset M\}$. By construction, $\{J_M\} \cup (\cup_{P \in \mathscr{P}_M} P) \subseteq M$. Consider any $v \in M$. Either $v \in J(G)$, or there exists a $P \in \mathscr{P}(G)$ such that $v \in P$. If $v \in J(G)$ then clearly $v \in J_M$. If instead there exists $P \in \mathscr{P}(G)$ such that $v \in P$, then $P \cap M \neq \emptyset$. By Lemma \ref{lem:jp-contained}, $P \subset M$ so $P \in \mathscr{P}_M$. Thus, $M \subseteq (J_M \cup (\cup_{P \in \mathscr{P}_M} P))$. By inclusion, $M = J_M \cup (\cup_{P \in \mathscr{P}_M} P)$.

  We first show that Statement 1 holds. Let $P \in \mathscr{P}_M$. Since $M$ is assumed to be a fort neighborhood, by Lemma \ref{lem:jp-contained} $N(P) \subset N[P] \subseteq M$. By Proposition \ref{prop:num-jp-neighbors}, $N(P) \subseteq J(G)$. Thus Statement 1 is satisfied as for all $v \in N(P)$, $v \in J_M$. Next we show that Statement 2 holds. Since $M \in \mathscr{M}(G)$, there exists some $F \in \mathscr{F}(G)$ such that $N[F] = M$. For any $u \in J_N$ has a neighbor not in $F$ or $N(F)$, so $u \in N(F)$. By inclusion, $J_n \subseteq N(F)$. Furthermore, $N[F] \backslash J_N = (N(F) \cup F) \backslash J_N = (N(F) \backslash J_N) \cup (F \backslash J_N) = (N(F) \backslash J_N) \cup F \supseteq F$. If we assume for that there exists $u \in J_N$ such that $|(N(u) \backslash J_N) \cap  M| \leq 1$, then since $F \subseteq N[F]\backslash J_N$,
  $$ |N(u) \cap F| \leq |N(u) \cap (N[F] \backslash J_N)| = |(N(u) \backslash J_N) \cap M| \leq 1. $$
  This is a contradiction as since $F$ is a fort and $u \in N(F)$, $|N(u) \cap F| \geq 2$. We conclude that $|(N(u) \backslash J_N) \cap M| \geq 2$ must hold for all $u \in J_N$, which is Statement 2. 

  Now we consider arbitrary $J_M \subseteq J(G)$, $\mathscr{P}_M \subseteq \mathscr{P}(G)$ such that Statements 1 and 2 hold. Let $M = J_M \cup (\cup_{P \in \mathscr{P}_M} P)$, and consider the vertex set $F = M \backslash J_N$. To complete our proof we must show that $M = N[F]$ and that $F \in \mathscr{F}(G)$. To see that $M = N[F]$, first consider an arbitrary $v \in N[F] = F \cup N(F)$. Since $F \subseteq M$, if $v \in F$ then clearly $v \in M$. If $v \in N(F)$, then $v$ must have a neighbor $u \in F$ such that either $u \in J_M \backslash J_N$ or $u \in P$ for some $P \in \mathscr{P}_M$. If $u \in P$ for some $P \in \mathscr{P}_M$, then by Statement 1, $v \in M$. If $u \in J_M \backslash J_N$, then $N(u) \subseteq M$ so $v \in M$. Thus all $v \in N[F]$ are in $M$, so $N[F] \subseteq M$

  Next, consider an arbitrary $v \in M$. Either $v \in J_M$, or there exists $P \in \mathscr{P}_M$ such that $v \in P$. In the latter case, $P \subset F$ thus $v \in F$. In the first case, either $v \in J_M \backslash J_N$ or $v \in J_N$. Additionally, if $v \in J_M \backslash J_N$, then $v \in F$. If instead $v \in J_N$ then by Statement 2, $|N(v) \cap F| = |N(v) \cap (M \backslash J_N)| = |(N(v) \backslash J_N) \cap M| \geq 2$, thus $v \in N(F)$. In all cases, $v \in N[F]$ so $M \subseteq N[F]$, and thus by inclusion $M = N[F]$.

  Lastly, the only case in which $v \in M$ and $v \in N(F)$ was the case in which $v \in J_N$. In this case we showed that $|N(v) \cap F| \geq 2$, thus by definition $F$ is a fort. 
\end{proof}

Using the conditions presented in Theorem \ref{thm:NS_FN}, some structural features of fort neighborhoods can be shown. In particular minimal fort neighborhoods must induce connected subgraphs.

\begin{theorem} \label{thm:M-conn}
  For any graph $G$ and any $M \in \mathscr{M}^0(G)$, $G[M]$ is connected.
\end{theorem}

\begin{proof}
  This statement can be shown by its contrapositive. Let $G$ be a graph, $M \in \mathscr{M}(G)$ such that $G[M]$ has at least two components with vertex sets $M_1$, $M_2$, respectively. We proceed by showing that $M_1$ satisfies the conditions laid out in Theorem \ref{thm:NS_FN}, and thus $M_1 \in \mathscr{M}(G)$. Let $J_{M_1} = M_1 \cap J(G)$, $P_{M_1} = \{P \in \mathscr{P}(G) : P \subset M_1\}$. It is straightforward, but tedious, to show that for any $P \in \mathscr{P}(G)$ contained in $M$, either $P \cap M_1 = \emptyset$ or $P \subset M_1$ is true, and thus $M_1 = J_{M_1} \cup ( \cup_{P \in \mathscr{P}_{M}} P)$. Since each $P \in \mathscr{P}_{M_1}$ is also contained in $M$, by Theorem \ref{thm:NS_FN}, any vertex $v \in N(P)$ must be in $M$ as well. Since $v \in N(P) \cap M$, and $P \subset M_1$, $v$ must also be in $M_1$. Thus, $M_1$ satisfies Statement 1 of Theorem \ref{thm:NS_FN}.

  Now focusing on Statement 2, consider $J_{N} := \{v \in J_{M}: N(v) \not \subseteq M\}$, and $J_{N_1} := \{v \in J_{M_1} : N(v) \not \subseteq M_1\}$. For any $v \in J_{N_1}$, $v \in J(G)$ and $v \in M$. Additionally, $v$ must be adjacent to a vertex $u \not \in M$, as any vertex in $M$ and adjacent to $v$ must also be in $M_1$. Thus $v \in J_N$. By Statement 2 of Theorem \ref{thm:NS_FN}, $|N(v)\backslash J_N) \cap M| \geq 2$. Again, every neighbor of $v$ that is in $M$ must also be in $M_1$, so $|(N(v) \backslash J_{N_1}) \cap M_1| \geq |(N(v) \backslash J_N) \cap M_1| = |(N(v) \backslash J_N) \cap M| \geq 2$. Since $M_1$ also satisfies Statement 2, by Theorem \ref{thm:NS_FN} $M_1$ is a fort neighborhood in $G$. Thus if $G[M]$ has more than component, $M$ is not a minimal in $\mathscr{M}(G)$. 
\end{proof}

As shown in the proof of Theorem \ref{thm:M-conn}, each component induced by a fort neighborhood is itself a fort neighborhood. We include this as the following corollary. 

\begin{corollary}
For any graph $G$, $M \in \mathscr{M}(G)$, if $C$ is a component in $G[M]$ then $V(C) \in \mathscr{M}(G)$.
\end{corollary}

\begin{figure}
\centering
\begin{tikzpicture}[scale=.75]
\begin{scope}[every node/.style={circle,draw,fill=white!40,minimum size=2mm}]
    \node (a) at (0,0) {};
    \node (b) at (0,3) {};
    \node (c) at (3,3) {};
    \node (d) at (3,0) {};
\end{scope}
\begin{scope}[every node/.style={circle,draw,fill=white!40,minimum size=2mm}]
    \node (ab) at (0,1.5) {};
    \node (ac) at (2.25,2.25) {};
    \node (ad) at (1.5,0) {};
    \node (bc) at (1.5,3) {};
    \node (cd) at (3,1.5) {};
    \node (bd) at (.75,2.25) {};
    \node (aa) at (-.75,-.75) {};
    \node (bb) at (-.75,3.75) {};
    \node (cc) at (3.75,3.75) {};
    \node (dd) at (3.75,-.75) {};
\end{scope}
\begin{scope}
  \draw[very thick] (a) to (aa)
  (a) to (ab)
  (a) to (ac)
  (a) to (ad)
  (b) to (bc)
  (b) to (bd)
  (c) to (cd)
  (b) to (ab)
  (b) to (bb)
  (d) to (cd)
  (c) to (cc)
  (d) to (dd)
  (c) to (bc)
  (c) to (ac)
  (d) to (ad)
  (d) to (bd);
\end{scope}
\end{tikzpicture}
\hspace{1cm}
\begin{tikzpicture}[scale=.75]    
\begin{scope}[every node/.style={circle,draw,fill=white!40,minimum size=2mm}]
	\node (11) at (0,	2) {};
	\node (12) at (1.175570505,	1.618033989) {};
	\node (22) at (1.902113033,	0.6180339887) {};
  \node (23) at (1.902113033,	-0.6180339887) {};
	\node (33) at (1.175570505,	-1.618033989) {};
	\node (34) at (0,	-2) {};
	\node (44) at (-1.175570505,-1.618033989) {};
	\node (45) at (-1.902113033,	-0.6180339887) {};
	\node (55) at (-1.902113033,	0.6180339887) {};
	\node (51) at (-1.175570505,	1.618033989) {};
    
  \node (1) at (0,	3) {};
	\node (2) at (2.853169549,	0.9270509831) {};
	\node (3) at (1.763355757,	-2.427050983) {};
	\node (4) at (-1.763355757,	-2.427050983) {};
	\node (5) at (-2.853169549,	0.9270509831) {};
    
	\node (13) at (0.3090169944,	0.9510565163) {};
	\node (24) at (1,	0) {};
	\node (35) at (0.3090169944,	-0.9510565163) {};
  \node (41) at (-0.8090169944,	-0.5877852523) {};
	\node (52) at (-0.8090169944,	0.5877852523) {};
    
  \draw[very thick] (1) to (11)
  (2) to (22)
  (3) to (33)
  (4) to (44)
  (5) to (55)

  (11) to (12)
  (12) to (22)
  (22) to (23)
  (23) to (33)
  (33) to (34)
  (34) to (44)
  (44) to (45)
  (45) to (55)
  (55) to (51)
  (51) to (11)
    
  (11) to (13)
  (13) to (33)
  (22) to (24)
  (24) to (44)
  (33) to (35)
  (35) to (55)
  (44) to (41)
  (41) to (11)
  (55) to (52)
  (52) to (22);
\end{scope}
\end{tikzpicture}
\caption{Examples of construction used in Proposition \ref{prop:fn-count}, for $k=4,5$}
\label{fig:fn-bound}
\end{figure}
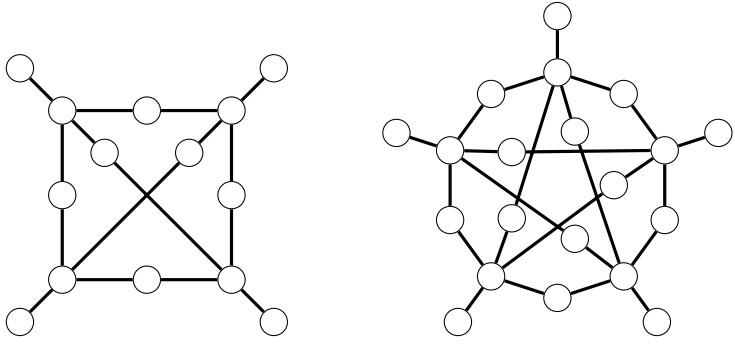

To complete this section, we now provide a proof for Proposition \ref{prop:fn-count}. \newline

\noindent
\textbf{Proposition \ref{prop:fn-count}.} \textit{\propThree}

\begin{proof}
  We verify this proposition with direct construction. Let $k \geq 3$, and consider the graph $G_k$ generated by the following procedure: let $G_k$ be the complete graph on $k$ vertices, for each edge $v_iv_j \in E$, delete the edge $v_iv_j$ and add vertex $v_{ij}$ to $V$ and edges $v_iv_{ij}, v_{ij}v_j$ to $E$. Finally, append a leaf $u_i$ to each vertex $v_i$ that was originally contained in $V(G_k)$. Examples of this construction are given in Figure \ref{fig:fn-bound}.

  Next, consider $M \subset V$ such that $M$ induces a cycle in $G_k$. By construction, for any $v \in V, \deg(v) \in \{1, 2, k\}$. Since no two vertices of the same degree are adjacent in $G_k$, $M$ must be made up of alternating degree 2 and degree $k$ vertices. Select $v_{ij} \in M$ such that $\deg(v_{ij}) = 2$ and update $M$ by removing $v_{ij}$ and adding leaves $u_i$ and $u_j$ to $M$. $J_M = \{v \in M : \deg(v) \geq 3\}$ and $\mathscr{P}_M = \{\{v\} : v \in M, \deg(v) < 3\}$ satisfy the conditions of Theorem \ref{thm:NS_FN}, so $M \in \mathscr{M}(G_k)$. Equivalently, $F = \{v \in M : \deg(v) < 3\} \in \mathscr{F}(G_k)$ and $N[F] = M$. 

  We show that $M$ is a minimal fort neighborhood by contradiction. Assume that $M$ is not minimal in $\mathscr{M}(G_k)$. Then there exists $M^0 \in \mathscr{M}(G_k)$ such that $M^0 \subset M$ and $M^0$ is minimal in $\mathscr{M}(G_k)$. Since $G_k[M]$ is a path and $G[M^0]$ must be connected (by Theorem \ref{thm:M-conn}), $G[M^0]$ must also be a path. Since each $F \in \mathscr{F}(G_k)$ consists of at least 2 vertices, $|M^0| > 1$. $G_k[M^0]$ must have 2 leaves. If the leaves in $G_k[M^0]$ are also leaves in $G_k$, then $G_k[M]$ could not be a path since $M^0$ is a strict subset of $M$. This cannot be the case, so there must exist a vertex $v \in M^0$ which is a leaf in $G_k[M^0]$ but not in $G_k$. Since $M^0 \in \mathscr{M}(G_k)$ however, there must exist $J_{M^0} \subseteq J(G)$, $\mathscr{P}_{M^0} \subseteq \mathscr{P}(G_k)$ that satisfy Theorem \ref{thm:NS_FN} and partition $M^0$. Since $|N(v) \cap M| < |N(v)|$, $v$ cannot be contained in any $P \in \mathscr{P}_{M^0}$. Then it must be that case that $v \in J_{M^0}$. Since $N(v) \not\subseteq M^0$, $|(N(v) \backslash J_N) \cap M| \geq 2$ must hold if Statement 2 of Theorem \ref{thm:NS_FN} is satisfied. This leads to a contradiction however as $|(N(v) \backslash J_N) \cap M| \leq |N(v) \cap M| = 1$. $M^0$ cannot satisfy Theorem \ref{thm:NS_FN}, so we conclude that $M$ is a minimal fort neighborhood. $M \in \mathscr{M}^0(G_k)$. 

  A straightforward counting argument shows that for a given $k$, there are $\sum^k_{\ell=3} \binom{k}{\ell} \frac{\ell!}{2}$ distinct choices for $M$; $G_k$ contains $\binom{k}{\ell}\frac{(\ell-1)!}{2}$ distinct cycles of length $2\ell$ for $\ell \in \{3, \dots, k\}$ and for a cycle of length $2\ell$ there are $\ell$ possible choices of $M$. Each choice of $M$ is a distinct minimal fort neighborhood, thus each $M$ is contained in $\mathscr{M}^0(G_k)$. Therefore, 
  $$
  |\mathscr{M}^0(G_k)|
  \geq \sum^k_{\ell=3} \binom{k}{\ell} \frac{\ell!}{2}
  = \frac{1}{2} \sum^k_{\ell=0} \left(\frac{k!}{(k-\ell)!}\right) - \frac{1}{2}(k^2+1).
  $$
  Making use of the factorial identity $\lim_{n \to \infty} \sum^n_{i=0} \frac{1}{i!} = e$,
  $$ \lim_{k \to \infty} \frac{|\mathscr{M}^0(G_k)|}{k!} \geq \lim_{k \to \infty} \frac{ \frac{1}{2} \sum^{k}_{\ell = 0} \left(\frac{k!}{(k-\ell)!}\right)}{k!} - \frac{\frac{1}{2} (k^2+1)}{k!} = \frac{e}{2} - 0 > 0. $$
Thus, $|\mathscr{M}^0(G_k)| = \Omega(k!)$ as $k \to \infty$. Since $n_k = \Theta(k^2)$, $|\mathscr{M}^0(G_k)| = \Omega((\sqrt{n_k})!)$, as $k \to \infty$. 
\end{proof}

For a large graph $G$, $|\mathscr{M}^0(G)|$ may also be very large. Proposition \ref{prop:fn-count} provides a lower bound on this number, but this bound may not be tight asymptotically. Instead, there may exist a sequence of graphs $\{G_k\}^\infty_{k=1}$ such that $|\mathscr{M}^0(G_k)| = \Omega(2^{n_k})$ as $k \to \infty$. However, with this bound we have shown that both $|\mathscr{M}^0(G_k)|$ grows asymptotically faster than polynomial and sub-exponential functions\footnote{By sub-exponential functions, we specifically mean the set of functions $O(2^{\sqrt{n}})$.}. Thus in general, enumerating all of the constraints considered in Model 1 can significantly hinder the runtime performance of computational methods. 

\section{Detection of Fort Neighborhoods} \label{sec:comp}

As defined in Section \ref{sec:structural}, Model 1 can be used to find the power domination number of a general graph $G$. In practicality, this model must be solved by constraint generation and depends on efficient methods for the detection of fort neighborhoods. As is standard in constraint generation strategies, a subset $\mathscr{M}_0 \subset \mathscr{M}(G)$ is generated and a fractional solution to the Linear Program (LP) relaxation of Model 1 is found. Then, a violated constraint in the form of (\ref{cons:1}) is identified and the corresponding fort neighborhood $M$ is added to $\mathscr{M}_0$.

In this section, we discuss the computational complexity of detecting these violated constraints, and provide two computational methods by which violated constraints can be practically identified.

\subsection{Computational Complexity}
The minimum weighted fort neighborhood decision problem, \textsc{Min-M}, is stated as follows: \newline

\noindent
\textbf{Problem}: \textsc{Min-M} \\
\textbf{Instance}: Graph $G$, Function $w:V \to [0,1]$, Integer $k$ \\
\textbf{Question}: Does there exist $M \in \mathscr{M}(G)$ such that $\sum_{v\in M} w(v) < k$? \\

When solving Model 1 with a constraint generation strategy, \textsc{Min-M} is the subproblem that must be solved to find violated constraints. In this context, $k=1$ and the weight function $w$ is the map between vertices in $V$ and their values in a fractional solution to the master problem.

\begin{theorem}
\textsc{Min-M} is in $\mathcal{NP}$.
\end{theorem}
\begin{proof}
  To show \textsc{Min-M} is in $\mathcal{NP}$, consider an arbitrary $\textsc{Min-M}$ instance $\left<G, w, k\right>$ and candidate solution $M \subseteq V$. In linear time it can be determined if $\sum_{v \in M} w(v) < k$ holds. In order to determine if $M \in \mathscr{M}(G)$, we consider the sets $J_M := M \cap J(G)$, $\mathscr{P}_M := \{P \in \mathscr{P}(G) : P \subseteq M\}$, verify that $M = \{J_M\} \cup (\cup_{P \in \mathscr{P}_M} P)$ and that Statements 1 and 2 of Theorem \ref{thm:NS_FN} hold. If $M \neq \{J_M\} \cup (\cup_{P \in \mathscr{P}_M} P)$, then by Lemma \ref{lem:jp-contained} $M \not \in \mathscr{M}$. If either Statement 1 or 2 does not hold then by Theorem \ref{thm:NS_FN}, $M \not \in \mathscr{M}$. In either of these cases $\left< G, w, k\right>$ is determined to be a '$no$' instance. If instead the lemma is satisfied and both statements hold, then $M$ satisfies the conditions of Theorem \ref{thm:NS_FN} and thus $M \in \mathscr{M}$ so $\left<G, w, k \right>$ can be evaluated as a '$yes$' instance.

  The sets $J(G), \mathscr{P}(G), J_M, \mathscr{P}_M$ can all be calculated in linear time, and both Statements 1 and 2 of Theorem \ref{thm:NS_FN} can also be verified in linear time. Thus $\left< G, w, k\right>$ can always be evaluated in a polynomial number of steps. 
\end{proof}

We next state a related problem in which a given vertex is required to be in the identified fort neighborhood. We regard this problem as \textsc{Restricted-Min-M}, the restricted minimum weight fort neighborhood problem. \\

\noindent
\textbf{Problem}: \textsc{Restricted-Min-M} \\
\textbf{Instance}: Graph $G$, Vertex $v$, Function $w:V \to [0,1]$, Integer $k$ \\
\textbf{Question}: Does there exist $M \in \mathscr{M}(G)$ such that $\sum_{v\in M} w(v) < k, v \in M$? \\

This second problem is also interesting from a complexity standpoint as it may be possible to create a heuristic for \textsc{Min-M} that finds violated constraints using a local search protocol or by iteratively growing a vertex set. Unfortunately, the following theorem shows that such methods that solve this alternative problem exactly may scale poorly with problem size, due to the complexity of \textsc{Restricted-Min-M}.

\begin{theorem} \label{thm:fort-comp}
\textsc{Restricted-Min-M} is NP-Complete. 
\end{theorem}
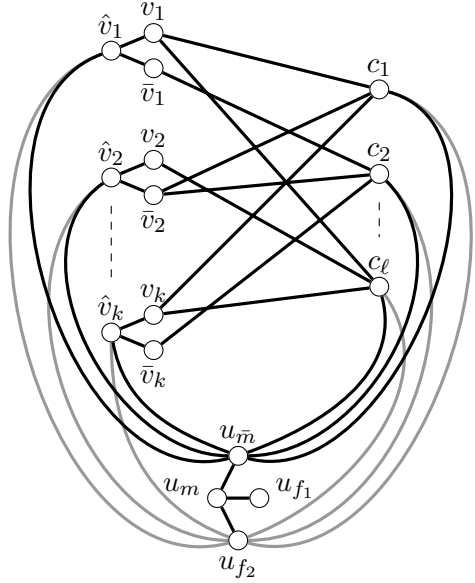
\begin{figure}
  \centering
  \begin{tikzpicture}[scale=.75]
\def     \cxpos{2} 
\def    \cxposi{\cxpos   + 3}
\def   \cxposii{\cxposi  + .875}
\def  \cxposiii{\cxposii + 1.5}

\def \cyposi{0} 
\def \cyposii{-1.5}
\def \cyposiii{-3.5}

\def \vxpos{-2} 
\def \vhatxpos{\vxpos - .75}

\def \vyposj{1} 
\def \vyposjj{-1.25}
\def \vyposjjj{-4}

\def \vbaryposj{\vyposj - .625}
\def \vbaryposjj{\vyposjj - .625}
\def \vbaryposjjj{\vyposjjj - .625}

\def \vhatyposj{\vyposj - .3125}
\def \vhatyposjj{\vyposjj - .3125}
\def \vhatyposjjj{\vyposjjj - .3125}

\def \vhatxposi{\vhatxpos - 3}
\def \vhatxposii{\vhatxposi - .875}
\def \vhatxposiii{\vhatxposii - 1.5}

\def \wxpos{-0.5}
\def \wypos{-6.5}
\def \wmanxpos{\wxpos - .375}
\def \wmanypos{\wypos - .75}
\def \wnautxpos{\wmanxpos + .75}
\def \wnautypos{\wypos - .75}
\def \wforbxpos{\wxpos}
\def \wforbypos{\wypos - 1.5}
\def \opa{.4}

\begin{scope}[every node/.style={circle,draw,fill=white!40,minimum size=2.5mm, inner sep=0pt, label distance = .025cm}]
    \node[label={[rectangle]:$c_1$}] (c1)   at (\cxpos, \cyposi) {};
    
    \node[label={[rectangle]:$c_2$}] (c2)   at (\cxpos, \cyposii) {};
    
    \node[label={[rectangle]:$c_{\ell}$}] (c3)   at (\cxpos, \cyposiii) {};
    
    \node[label={[rectangle]:$v_1$}] (v1)                                at (\vxpos,    \vyposj) {};
    \node[label={[rectangle, yshift=-.05cm]below:$\bar v_{1}$}] (v1-bar) at (\vxpos,    \vbaryposj) {};
    \node[label={[rectangle]:$\hat v_{1}$}]   (v1-hat)                   at (\vhatxpos,  \vhatyposj) {};
    
    \node[label={[rectangle]:$v_2$}] (v2)                                at (\vxpos,    \vyposjj) {};
    \node[label={[rectangle, yshift=-.05cm]below:$\bar v_{2}$}] (v2-bar) at (\vxpos,    \vbaryposjj) {};
    \node[label={[rectangle]:$\hat v_{2}$}]      (v2-hat)                at (\vhatxpos,  \vhatyposjj) {};
    
    \node[label={[rectangle]:$v_k$}] (v3)                                at (\vxpos,    \vyposjjj) {};
    \node[label={[rectangle, yshift=-.05cm]below:$\bar v_{k}$}] (v3-bar) at (\vxpos,    \vbaryposjjj) {};
    \node[label={[rectangle]:$\hat v_{k}$}]      (v3-hat)                at (\vhatxpos,  \vhatyposjjj) {};
    
    \node[label={[rectangle]:$u_{\bar{m}}$}] (u-out) at (\wxpos, \wypos) {};
    \node[label={[rectangle, xshift=-.05cm, yshift=.15cm]left  : $u_m$}] (u-man) at (\wmanxpos, \wmanypos) {};
    \node[label={[rectangle, xshift= .05cm, yshift=.125cm]right : $u_{f_1}$}] (f-2) at (\wnautxpos, \wnautypos) {};
    \node[label={[rectangle, yshift=-.05cm]below : $u_{f_2}$}] (f-1) at (\wforbxpos, \wforbypos) {};
\end{scope}

\begin{scope}

    \draw[very thick] (v1-hat) to (v1)
                      (v1-hat) to (v1-bar)
                      (v2-hat) to (v2)
                      (v2-hat) to (v2-bar)
                      (v3-hat) to (v3)
                      (v3-hat) to (v3-bar)
                      
                      (c1) to (v1)
                      (c1) to (v2-bar)
                      (c1) to (v3)
                      (c2) to (v1-bar)
                      (c2) to (v2-bar)
                      (c2) to (v3-bar)
                      (c3) to (v1)
                      (c3) to (v2)
                      (c3) to (v3)
                      
                      (u-out) to (u-man)
                      (u-man) to (f-1)
                      (u-man) to (f-2);
                      
    \path[very thick] (v1-hat) edge[out=200,   in=195] (u-out);
    \path[very thick] (v2-hat) edge[out= 215,  in=180] (u-out);
    \path[very thick] (v3-hat) edge[out= 280,  in=160] (u-out);
    \path[very thick] (c1)     edge[out=-15,   in=-15] (u-out);
    \path[very thick] (c2)     edge[out=-45,   in=0]   (u-out);
    \path[very thick] (c3)     edge[out=-75,   in= 20]  (u-out);
    
    \path[very thick, opacity = \opa] (v1-hat) edge[out=200,   in=195] (f-1);
    \path[very thick, opacity = \opa] (v2-hat) edge[out= 215,  in=180] (f-1);
    \path[very thick, opacity = \opa] (v3-hat) edge[out= 270,  in=160] (f-1);
    \path[very thick, opacity = \opa] (c1)     edge[out=-15,   in=-15] (f-1);
    \path[very thick, opacity = \opa] (c2)     edge[out=-45,   in=0]   (f-1);
    \path[very thick, opacity = \opa] (c3)     edge[out=-55,  in= 20]  (f-1);

    \draw[thin, dashed, shorten >=15pt, shorten <=7pt] (c2) -- (c3);
    \draw[thin, dashed, shorten >=15pt, shorten <=7pt] (v2-hat) -- (v3-hat);
\end{scope}

\end{tikzpicture}
  \caption{Example of the construction used in Theorem \ref{thm:fort-comp}.}
\label{fig:fort-comp}
\end{figure}
\begin{proof}
  To show \textsc{Restricted-Min-M} is in $\mathcal{NP}$ consider an arbitrary instance, $\left<G, v, w, k \right>$, and a candidate solution $M$. Then $\left<G, v, w, k \right>$ is a '$yes$' instance of \textsc{Restricted-Min-M} if and only if $\left< G, w, k \right>$ is a '$yes$' of \textsc{Min-M} and $v \in M$. Since each of these can be determined in a polynomial number of operations, \textsc{Restricted-Min-M} is in $\mathcal{NP}$.

  To show that \textsc{Restricted-Min-M} is also $\mathcal{NP}$-hard, we provide a polynomial reduction to the well known \textsc{3-SAT} problem. Consider an instance of \textsc{3-SAT}, $\left<f\right>$, where $f$ is a boolean formula in conjunctive normal form, with variables $V_1, \dots, V_k$ and clauses $C_1, \dots, C_\ell$ each consisting of 3 variable assignments. Then let $G' = (V', E')$, where $V' := (\cup_{i=1}^k \{v_i, \bar{v}_i, \hat{v}_i\}) \cup \{c_i\}_{i=1}^\ell \cup \{u_m, u_{\bar{m}}, u_{f_1}, u_{f_2}\}$ and $E' := \{v_ic_j : V_i \in C_j\} \cup \{\bar{v}_i c_j : \bar{V}_i \in C_j\} \cup (\cup^k_{i=1} \{\hat{v}_iv_i, \hat{v}_i \bar{v}_i, \hat{v}_iu_{\bar{m}}, \hat{v}_iu_{f_2}\}) \cup (\cup^\ell_{i=1} \{c_iu_{\bar{m}}, c_iu_{f_2}\}) \cup \{u_{\bar{m}}u_m, u_mu_{f_2}, u_mu_{f_1}\}$. The constructed graph $G'$ contains $\ell+3k+4$ vertices and $5\ell+4k+3$ edges; it can be constructed in a number of steps which is linear with respect to the size of the input $f$. A visualization for this construction is provided as Figure \ref{fig:fort-comp}. In addition, we define a weighting function $w':V(G') \to [0,1]$ such that
  $$ w'(v) =
  \begin{cases}
    \begin{aligned}
    k+1 &\text{ if } v = u_{\bar{m}} \\
    1   &\text{ if } v = v_i \text{ or } v = \bar{v}_i \text{ for some } i \:\: .\\
    0   &\text{ else}
    \end{aligned}
  \end{cases} $$
  We complete the proof by showing that $f$ is satisfiable if and only if $\langle G', u_m, w',$ $k+1 \rangle$ is '$yes$' instance of \textsc{Restricted-Min-M}. First, assume that $f$ is satisfiable, with some variable assignment $\dot{V}_1, \dots \dot{V}_k$ where $\dot{V}_i$ is either equal to $V_i$ or $\bar{V}_i$ for $i \in \{1, \dots, k\}$. Then let $F := \{u_{f_1}, u_{f_2}\} \cup \{v_i : \dot{V}_i = V_i\} \cup \{\bar{v}_i : \dot{V}_i = \bar{V}_i\}$. It is straight forward to show that $F$ is a fort, and has closed neighborhood $M := N[F] = F \cup \{u_m\} \cup \{\hat{v}_i\}_{i=1}^k \cup \{c_i\}_{i=1}^\ell$. Thus $M \in \mathscr{M}(G')$, $u_m \in M$, and $\sum_{v \in M} w'(v) = k < k+1$, so $\left<G', u_m, w', k+1 \right>$ is indeed a '$yes$' instance of \textsc{Restricted-Min-M}.

  Alternatively, assume that $\left< G', u_m, w', k+1 \right>$ is a '$yes$' instance of \textsc{Restr-icted-Min-M}. Then there exists $M \in \mathscr{M}(G)$, $F \in \mathscr{F}(G)$ such that $M = N[F]$, $u_m \in M$, $\sum_{v \in M} w'(v) < k+1$. Since $w'(w_{\bar{m}}) = k+1$, $w_{\bar{m}} \not \in M$ so $u_m \in N(F)$. Since $u_m \in N(F)$, $u_m$ must have at least two neighbors in $F$, thus $u_{f_1}, u_{f_2} \in F$. For each $i$ then, $\hat{v}_i \in N[F]$. Like $u_m$, since each $\hat{v}_i$ is also adjacent to $u_{\bar{m}}$, $\hat{v}_i$ cannot be in $F$ itself and instead must be contained in $N(F)$. Since each vertex in $N(F)$ must have two neighbors in $F$, for each $i$, at least one of $\hat{v}_i$'s two remaining neighbors, namely $v_i$ and $\bar v_i$, must be in $F$. By the pigeonhole principle, if any pair $v_i, \bar v_i$ are both in $F$ then since $M$ is assumed to have total weight less than or equal to $k$, then there exists a pair $v_j, \bar v_j$ such that neither vertex is contained in $F$. This cannot be the case so we see that for any $i \in \{1, \dots, k\}$, either $v_i \in F$ or $\bar v_i \in F$. Using $F$ then, we can construct a variable assignment $\dot V_1, \dots, \dot V_k$ such that for each $i \in \{1, \dots, k\}$, $\dot V_i = V_i$ if $v_i \in F$ or $\dot V_i = \bar V_i$ if $\bar v_i \in F$.

  Next, we consider the set of clause vertices, $c_1, \dots, c_\ell$. As before, since each clause vertex is adjacent to a vertex in $F$ and vertex not in $M$, $\{c_1, \dots, c_\ell\}$ $\subset N(F)$. Since each vertex in $N(F)$ must be adjacent to at least two vertices in $F$, each clause vertex must have at least one adjacent variable assignment vertex in $F$. Since by construction a clause vertex is adjacent to a variable assignment vertex in $G'$ if and only if that variable assignment literal appears in the corresponding clause in $f$, each clause $C_1$ is satisfied by at least one variable assignment literal in $\dot V_1, \dots, \dot V_k$. $f$ is satisfied by variable assignments $\dot V_1, \dots, \dot V_k$, so $f$ is satisfiable. 
\end{proof}

\subsection{Computational Methods}
In this section, we provide two approaches for separating violated constraints in Model 1. The first method we present finds fort neighborhoods of minimum weight through the use of an IP. Unfortunately this method scales poorly in larger instances. While few constraints have been incorporated into Model 1 many fort neighborhoods of may have weight 0. Without a tie breaking criterion, large vertex sets may be added to the master problem. When this occurs, a shallow violated constraint is added and little progress is made towards optimality. The second method we present avoids this by identifying minimum cardinality fort neighborhoods with total weight less than 1. We compare the performance of these methods in Figure \ref{tab:sep-results}, noting the improved performance of the cardinality based approach.

We also provide a technique for initializing the constraint set of Model 1. In practice, power grid graphs are often sparse and certain special structures can be identified quickly to generate tight constraints. 

\subsubsection{Minimum Weight Fort Neighborhoods}
Utilizing the characterization of fort neighborhoods given by Theorem \ref{thm:NS_FN}, an IP can be used to identify minimum weight fort neighborhoods in general weighted graphs. To simplify notation, given a graph $G$ and weight function $w:V \to [0,1]$, we define $w_v = w(v)$ for all $v \in J(G)$ and $w_P = \sum_{v \in P} w(v)$ for all $P \in \mathscr{P}(G)$. \newline

\noindent \text{\textbf{Model 2: Min Weight Fort Neighborhood}}
{\small
  \begin{align}
    \nonumber\min \:&\sum_{v \in J(G)} w_v M_v  + \sum_{P \in \mathscr{P}(G)} w_P F_P \\ 
    \text{s.t. }  \:&\sum_{v \in J(G)} M_v + {\sum_{P \in \mathscr{P}(G)} F_P \geq 1} \label{cons:2} \\
        & F_v \leq M_u  & \forall v \in J(G), \forall u \in J(G): v \in N[u]           \label{cons:3} \\
        & F_v \leq F_P  & \forall v \in J(G), \forall P \in \mathscr{P}(G): v \in N(P) \label{cons:4} \\
				& F_P \leq M_v  & \forall P \in \mathscr{P}(G), \forall v \in N(P)                \label{cons:5} \\
        & 2(M_v - F_v) \leq \sum_{\substack{u \in J(G): \\ v \in N(u)}} F_u + \sum_{P \in \mathscr{P}(G)} |N(v) \cap P| \: F_P  & \forall v \in J(G) \label{cons:6} \\
    \nonumber & F_P, F_v, M_v \in \{0,1\} & \forall v \in J(G), \forall P \in \mathscr{P}(G)		
\end{align}}

We formally prove the correctness of Model 2 in Theorem \ref{thm:FN-IP}, but first we provide the intuition behind the model's variables and constraints. As shown by Lemma \ref{lem:jp-contained}, fort neighborhoods can be partitioned into a set of junctions and junction paths. Variables $M_v$ indicate the junctions in a candidate fort neighborhood $M$, and variables $F_v$, $F_P$ indicate the junctions and junction paths in the interior of $M$, which we show is a fort. Since Theorem \ref{thm:NS_FN} requires junction paths contained in fort neighborhoods to be in the interiors of those fort neighborhoods, we do not need to include corresponding $M_P$ variables. The candidate fort neighborhood identified by the model, $M$, is indicated by the $M_v$ and $F_P$ variables. Thus, feasible solutions to Model 2 correspond to members of $\mathscr{M}(G)$, while the model's objective function ensures that optimal solutions are minimum weight members of $\mathscr{M}(G)$. 

Constraint (\ref{cons:2}) ensures $M$ is non-empty. Constraint (\ref{cons:3}) ensures junctions indicated to be in the interior of $M$ are also indicated to be in $M$. Along with Constraint (\ref{cons:4}), Constraint (\ref{cons:3}) also ensures that each junction in the interior of $M$ is not adjacent to any junction or junction path not in $M$. Constraint (\ref{cons:5}) ensures that junctions adjacent to any junction path in $M$ are also contained in $M$, and Constraint (\ref{cons:6}) ensures that any junction in $M$, but not the interior of $M$, is adjacent to at least two vertices in the interior of $M$.

\begin{theorem}\label{thm:FN-IP}
For any graph $G$, weight function $w$, and optimal solution of Model 2, $M = \{v \in J(G) : M_v = 1\} \cup \{v \in P : P \in \mathscr{P}(G), F_P = 1\}$ is a minimum weight fort neighborhood in $G$. 
\end{theorem}

\begin{proof}
  To verify this theorem it suffices to show that for each feasible solution for Model 2, $M = \{v \in J(G) : M_v = 1\} \cup \{v \in P : P \in \mathscr{P}(G), F_P = 1\} \in \mathscr{M}(G)$, and that for each $M \in \mathscr{M}(G)$, there exists a feasible solution for Model 2 such that vertices in $M$ are indicated by the $M_v$ and $F_P$ variables. Since the objective function of Model 2 is simply the total weight of $M$, a feasible solution is optimal if and only if it corresponds to a minimum weight fort neighborhood in $G$.

  First, let $\{M_v\}_{v \in J(G)}, \{F_v\}_{v \in J(G)}, \{F_P\}_{P \in \mathscr{P}(G)}$ be a feasible solution for Model 2, and define $J_M = \{v \in J(G) : M_v = 1\}$, $\mathscr{P}_M = \{P \in \mathscr{P}(G) : F_P = 1\}$, and $M = J_M \cup (\cup_{P \in \mathscr{P}_M} P)$. We prove that $M \in \mathscr{M}(G)$ by showing that $M, J_M,$ and $\mathscr{P}_M$ satisfy the sufficient conditions laid out in Theorem \ref{thm:NS_FN}. Since Constraint (\ref{cons:2}) is satisfied $J_M$, $\mathscr{P}_M$ cannot both be empty, so $M$ cannot be empty. Additionally, for any $v \in J(G)$ such that $v \in N(P)$, since Constraint (\ref{cons:5}) is satisfied: $M_v \geq F_P = 1$. Thus $v \in J_M$. The first condition of Theorem \ref{thm:NS_FN} is satisfied.

  Next we verify the second condition of Theorem \ref{thm:NS_FN}. Consider the set $J_N = \{v \in J_M : N(v) \not\subseteq M\}$. $J_N \subseteq J_M$, so for any $v \in J_N$, $M_v = 1$ must hold. Since $v \in J_N$, there must also exist a vertex $u \in N(v)$ such that $u \not \in M$. Either $u \in J(G)$ or there exists some $P' \in \mathscr{P}(G)$ such that $u \in P$. If $u \in J(G)$ then since Constraint (\ref{cons:3}) is satisfied and $M_u = 0$, $F_v \leq M_u = 0$. If instead there exists $P' \in \mathscr{P}(G)$ such that $u \in P'$ then since $F_{P'} = 0$ and Constraint (\ref{cons:4}) is satisfied, $F_v \leq F_{P'} = 0$. In all cases, if $v \in J_N$ then $M_v = 1$ and $F_v = 0$ must hold. Using this, we prove an identity for all $v \in J(G)$,
  $$ |(N(v) \backslash J_N) \cap M| = \sum_{\substack{u \in J(G): \\ v \in N(u)}} F_u + \sum_{P \in \mathscr{P}(G)} |N(v) \cap P| \: F_P. $$
  Since $J_M$ and each set in $\mathscr{P}_M$ are disjoint the following holds.
  \begin{equation*}
    \begin{split}
      |(N(v) \backslash J_N) \cap M| =& |(N(v) \backslash J_N) \cap (J_M \cup (\cup_{P \in \mathscr{P}_M} P))| \\
      =& |(N(v) \backslash J_N) \cap J_M| + |(N(v) \backslash J_N) \cap (\cup_{P \in \mathscr{P}_M} P)| \\
      =& |N(v) \cap (J_M \backslash J_N)| + |N(v) \cap (\cup_{P \in \mathscr{P}_M} P)|. 
    \end{split}
  \end{equation*}
  Additionally for any $u \in J(G)$ such that $F_u = 1$ holds, since Constraint (\ref{cons:3}) is satisfied and $u \not\in J_N$, $u$ must be in $J_M \backslash J_N$. Thus,
  $$ |N(v) \cap (J_M \backslash J_N)| = \sum_{u \in N(v) \cap J(G)} F_u = \sum_{\substack{u \in J(G): \\ v \in N(u)}} F_u. $$  
  Then since each $P \in \mathscr{P}(G)$ is disjoint,
  $$ |N(v) \cap (\cup_{P \in \mathscr{P}_M} P)| = \sum_{P \in \mathscr{P}_M} |N(v) \cap P| = \sum_{P \in \mathscr{P}(G)} |N(v) \cap P| \: F_P. $$
  Returning to our first equation, we see that if $v$ is also assumed to be in $J_N$, then Constraint (\ref{cons:6}) ensures that
  $$ |(N(v) \backslash J_N) \cap M| = \sum_{\substack{u \in J(G): \\ v \in N(u)}} F_u + \sum_{P \in \mathscr{P}(G)} |N(v) \cap P| \: F_P \geq 2(M_v - F_v) = 2. $$
Thus the second condition of Theorem \ref{thm:NS_FN} is also satisfied. 

Next we show that for any fort neighborhood $M \in \mathscr{M}(G)$, there exists a corresponding feasible solution for Model 2. Consider any $M \in \mathscr{M}(G)$ and for any $v \in J(G)$, let $M_v = 1$ if $v \in M$ and let $F_v = 1$ if $N[v] \subseteq M$. For all $P \in \mathscr{P}(G)$, let $F_P = 1$ if $P \subseteq M$. Lastly, set any unassigned $M_v, F_v, F_P = 0$. Since $M$ cannot be empty, Constraint (\ref{cons:2}) is satisfied. For Constraints (\ref{cons:3}), (\ref{cons:4}), and (\ref{cons:5}), we note that if the left hand side is equal to 0, then the constraint holds vacuously. It suffices then to consider only the cases which the left hand side is equal to 1. By construction, for any $v \in J(G)$, $F_v = 1$ if and only if $N[v] \subseteq M$. For any $u \in J(G): v \in N[u]$, $u$ must be in $M$ so $M_u = 1$. Thus, Constraint (\ref{cons:3}) must hold. Similarly, for any $w \in N(v)$ such that there exists a junction path $P$ containing $w$, $w$ must be in $M$. By Lemma \ref{lem:jp-contained}, since $P \cap M \neq \emptyset$, $P \subseteq M$. Thus $F_P = 1$, so Constraint (\ref{cons:4}) must also be satisfied. Next, consider any $P \in \mathscr{P}(G)$ such that $F_P = 1$, and $v \in N(P)$. By construction, $P \subseteq M$. By the first condition of Theorem \ref{thm:NS_FN}, $v \in M$ so $M_v = 1$; Constraint (\ref{cons:5}) holds.

Lastly, we consider Constraint (\ref{cons:6}). Let $v \in J(G)$. In the case that $(M_v-F_v) = 0$, then since the right hand side is always non-negative the constraint must hold. Thus, we need only consider the case that $M_v = 1$ and $F_v = 0$. In this case, by the way $M_v, F_v$ were constructed, $v \in M$ and $N(v) \not\subseteq M$. We note that for any $u \in J(G)$, $u \in J_N$ if and only if $M_u - F_u = 1$ (equivalently, $M_u = 1$ and $F_u = 0$). By the second condition of Theorem \ref{thm:NS_FN}, $|(N(v) \backslash J_N) \cap M| \geq 2$. Thus, by the counting identity shown before, 
$$ 2 \leq |(N(v) \backslash J_N) \cap M| = \sum_{\substack{u \in J(G): \\ u \in N(v)}} F_u + \sum_{P' \in \mathscr{P}(G)} |N(v) \cap P'| F_{P'}. $$
Constraint (\ref{cons:6}) holds, thus every $M \in \mathscr{M}(G)$ corresponds to a feasible solution for Model 2, with an objective value equal to the total weight of the vertices in $M$.
\end{proof}

\subsubsection{Minimum Cardinality Fort Neighborhoods}
As an alternative to identifying minimum weight fort neighborhoods, we next present Model 3. This model is a modification of Model 2 and identifies a minimum cardinality fort neighborhood with total weight less than 1, when model parameter $\epsilon$ is chosen to be a small positive number. \newline

\noindent \text{\textbf{Model 3: Min Cardinality Fort Neighborhood}}
{\small
  \begin{align}
    \nonumber\min \:&\sum_{v \in J(G)} M_v  + \sum_{P \in \mathscr{P}(G) |P|\:F_P} \\ 
    \nonumber\text{s.t. }  \:&\sum_{v \in J(G)} M_v + {\sum_{P \in \mathscr{P}(G)} F_P \geq 1}  \\
        & F_v \leq M_u  & \forall v \in J(G), \forall u \in J(G): v \in N[u] \nonumber  \\
        & F_v \leq F_P  & \forall v \in J(G), \forall P \in \mathscr{P}(G): v \in N(P) \nonumber \\
				& F_P \leq M_v  & \forall P \in \mathscr{P}(G), \forall v \in N(P) \nonumber    \\
        & 2(M_v - F_v) \leq \sum_{\substack{u \in J(G): \\ v \in N(u)}} F_u + \sum_{P \in \mathscr{P}(G)} |N(v) \cap P| \: F_P  & \forall v \in J(G) \nonumber \\
        &\sum_{v \in J(G)} w_v M_v  + \sum_{P \in \mathscr{P}(G)} w_P F_P \leq 1-\epsilon \label{cons:7} \\
    \nonumber & F_P, F_v, M_v \in \{0,1\} & \forall v \in J(G), \forall P \in \mathscr{P}(G)		
\end{align}}

Since the constraints enforced in Model 3 include all of the constraints given in Model 2, any feasible solution of Model 3 is also a feasible solution for Model 2 and thus by Theorem \ref{thm:FN-IP} corresponds to the vertex set of a fort neighborhood. Constraint (\ref{cons:7}) ensures that this fort neighborhood has total vertex weight less than 1. In a graph with vertex weights given by a fractional solution for our master problem, Model 1, any feasible solution for Model 3 provides a violated constraint for Model 1.

By minimizing the cardinality of the fort neighborhood, the inequality that is introduced to the master problem is more difficult to satisfy, thereby imposing a greater restriction on the set of feasible solutions. In our computational experiments, we found that this approach for separating out violated constraints outperformed the strategy using Model 2 as outlined in the previous section.

\subsubsection{Detecting Fort Neighborhoods of Special Forms} \label{sec:M-init}
In addition to adding constraints to Model 1 through the use of Model 2 and 3, fort neighborhoods of special forms can often be efficiently identified and included in Model 1 as an initial constraint set. While these forms of fort neighborhoods may do not appear in all graphs, they occur frequently in graph representations of power grids and can be exploited to benefit computation practice.
\begin{figure}
	\centering
  \hfill
\def \yoff {1}
\def \xoff {2}
\begin{tikzpicture}[scale=.75]
  \begin{scope}[every node/.style={circle,draw,minimum size=2mm}]
  \node[label={[rectangle]:}] (v1) at (0, \yoff) {};
  \node[label={[rectangle]:}] (v2) at (1*\xoff, \yoff) {};

  \node[label={[rectangle]:}] (v3)  at (0, -1*\yoff) {};
  \node[label={[rectangle]:}] (v4)  at (1*\xoff, -1*\yoff) {};
  \node[label={[rectangle]:$v$}] (v0) at (-.5*\xoff, 0) {};
  \node[label={[rectangle]:}] (v) at (-1*\xoff, 0) {};
  \end{scope}
  \begin{scope}
    \draw[dashed, very thick]
    (v1) to node[above]{$P_1$} (v2)
    (v3) to node[above]{$P_2$} (v4);
    \draw[very thick]
    (v)  to (v0)
    (v0) to (v1)
    (v0) to (v3);
  \end{scope}
\end{tikzpicture}
\hfill
\begin{tikzpicture}[scale=.75]
  \begin{scope}[every node/.style={circle,draw,minimum size=2mm}]
  \node[label={[rectangle]:}] (v1) at (0, \yoff) {};
  \node[label={[rectangle]:}] (v3)  at (0, -1*\yoff) {};
  \node[label={[rectangle]:$v$}] (v0) at (-.5*\xoff, 0) {};
  \node[label={[rectangle]:}] (v) at (-1*\xoff, 0) {};
  \end{scope}
  \begin{scope}
    \draw[dashed, very thick]
    (v1) to node[right]{$P_1$} (v3);
    \draw[very thick]
    (v)  to (v0)
    (v0) to (v1)
    (v0) to (v3);
  \end{scope}
\end{tikzpicture}
\hfill
\begin{tikzpicture}[scale=.75]
  \begin{scope}[every node/.style={circle,draw,minimum size=2mm}]
  \node[label={[rectangle]:}] (v1) at (0, \yoff) {};
  \node[label={[rectangle]:}] (v2) at (1*\xoff, \yoff) {};

  \node[label={[rectangle]:}] (v3)  at (0, -1*\yoff) {};
  \node[label={[rectangle]:}] (v4)  at (1*\xoff, -1*\yoff) {};
  \node[label={[rectangle]:$v$}] (v0) at (-.5*\xoff, 0) {};
  \node[label={[rectangle]:}] (u0) at (2*\xoff, 0) {};
  \node[label={[rectangle]:}] (v) at (-1*\xoff, 0) {};
  \node[label={[rectangle]:$u$}] (u) at (1.5*\xoff, 0) {};
  \end{scope}
  \begin{scope}
    \draw[dashed, very thick]
    (v1) to node[above]{$P_1$} (v2)
    (v3) to node[above]{$P_2$} (v4);
    \draw[very thick]
    (v)  to (v0)
    (v0) to (v1)
    (v0) to (v3)
    (u)  to (u0)
    (v2) to (u)
    (v4) to (u);
  \end{scope}
\end{tikzpicture}
\hfill
  \caption{\footnotesize Illustrations of fort neighborhoods of type I, II, and III (from left to right, respectively). In each example, $v,u$ are junction vertices and $P_1, P_2$ are junction paths. Fort neighborhoods of these special forms can be quickly detected and included as an initial constraint set for Model 1.}
    \label{fig:init-M}
\end{figure}
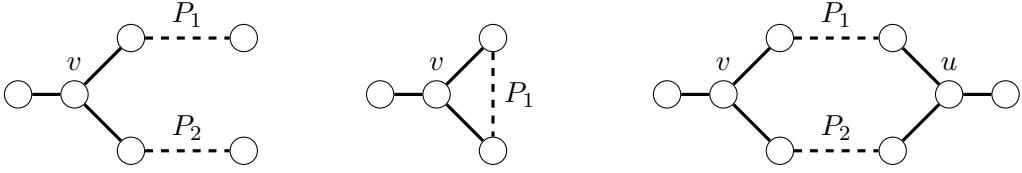

We consider fort neighborhoods of three types which can be easily described for a graph $G$ in terms of $J(G)$ and $\mathscr{P}(G)$. A fort neighborhood is of the first type if its vertices consist of a junction $v$, and two junction paths $P_1$, $P_2$ such that $N(P_1) = N(P_2) = \{v\}$. A fort neighborhood is of the second type if its vertices consist of a junction $v$ and a junction path $P$ such that $N(P) = \{v\}$ and $|N(v) \cap P| = 2$. Lastly, a fort neighborhood is of the third type if its vertices can be partitioned into junctions $v, u$ and junction paths $P_1, P_2$ such that $N(P_1) = N(P_2) = \{v, u\}$. Fort neighborhoods of these three special types are illustrated in Figure \ref{fig:init-M}.

In the computational experiments presented in Section \ref{sec:results}, we calculate $J(G), \mathscr{P}(G)$, and initialize the constraint set of Model 1 in a process similar to graph component identification. First, $J(G)$ is determined by iterating through the vertices of $G$. After $J(G)$ has been found, $\mathscr{P}(G)$ is determined by recording the vertex sets of the components of $H = G[V\backslash J(G)]$. When a junction path $P$ is detected, its neighborhood in $G$ is searched and fort neighborhoods of type II can be identified. If a fort neighborhood of type II is detected, then the corresponding constraint can be added to Model 1. If not, then since $P$ must either be adjacent to 1 or 2 junctions, the junctions $P$ is adjacent to can be recorded in an upper triangular $|J(G)|\times|J(G)|$ adjacency matrix like structure. If two junction paths are each adjacent to only one junction, then a fort neighborhood of type I is identified, and likewise if a pair junctions are adjoined by two distinct junction paths then a fort neighborhood of type III is identified. To prevent the accumulation of overlapping constraints, for any junction only the first fort neighborhood of type I or II is added to the initial constraint set of Model 1 and for any pair of junctions only the first fort neighborhood of type III is added.

We note that there are other structures that appear regularly in graph representations of electrical networks, and that these three special types are in no way exhaustive of easily recognizable fort neighborhoods. Possible directions of future work include the incorporation of additional commonly occurring yet small fort neighborhood structures.

\section{Computational Results} \label{sec:results}
In this section, we compare the runtime performance of the proposed method with existing methods. To the best knowledge of the authors, the IP proposed by Brimkov et al.~\cite{conn-pd} is the current state of the art computational methods for power domination, demonstrating superior runtime performance in comparison to the other so called infection model seen in power domination literature~\cite{aazami}. In the first part of this section, we reproduce the infection model proposed by Brimkov et al. \cite{conn-pd} and discuss a modification, motivated by Proposition \ref{prop:J-set}, which reduces the set feasible solutions permitted by the model. In the latter portion of this section, we compare the performance of both of these methods to the separation algorithm proposed in this paper.

\subsection{Power Domination Infection Model Modification} \label{sec:inf-model}
The IP introduced by Brimkov et al. \cite{conn-pd} identifies the power domination numbers of general graphs. To do so, at an optimal solution of the IP, the variables of the model correspond to both a minimum power dominating set as well as a sequence in which the initially uncolored vertices can be colored. For an arbitrary graph $G' = (V, E')$, the digraph $G = (V, E = \{uv : uv \in E(G')$ or $vu \in E(G')\})$ is constructed. The following IP is then defined using the auxiliary graph $G$. \newline

\noindent \text{\textbf{Model 4: Power Domination Infection Model} \cite{conn-pd}}
{\small
  \begin{align}
    \min \:         &\sum_{v \in V} s_v \nonumber \\ 
    \text{s.t. } \: & s_v + \sum_{uv \in E} y_{uv}  =  1   & \forall v \in V \nonumber \\
                    & x_u - x_v + (n+1) y_{uv}  \leq n   & \forall uv \in E \nonumber \\
                    & x_w - x_v + (n+1) y_{uv}  \leq n + (n+1)s_u & \forall uv \in E, \forall w \in N(u) \backslash \{v\} \nonumber \\
                    & x_v \in \{0, 1, \dots, n\}, s_v \in \{0,1\} & \forall v \in V, \nonumber \\
                    & y_{uv} \in \{0,1\}             & \forall uv \in E \nonumber
\end{align}}

In their paper, Brimkov et al. \cite{conn-pd} show that each power dominating set $S$ in $G'$ has a corresponding Model 4 feasible solution in which $S$ is indicated by the $s$ variables. In the case that $G'$ is connected and $\Delta(G') \geq 3$, this model can be modified by including an additional constraint,
{\small \begin{align}
    & s_v = 0 & \forall v \in V : \deg(v) \leq 2 \label{cons:8}.
  \end{align} }
Since Proposition \ref{prop:J-set} ensures that $G'$ has a minimum power dominating set contained in $J(G)$, Model 4 must have an optimal solution in which $s_v = 0$ for all $v$ with $\deg(v) \leq 2$. Thus when Constraint (\ref{cons:8}) is added to Model 4, the modified infection model remains feasible. Additionally, any feasible solution of the modified infection model is feasible for Model 4, so the optimum of the modified model is bounded below by the optimum of Model 4. By incorporating this modification however, the feasible region of both Model 4 and its linear relaxation are often significantly restricted. In the latter half of this section, we provide computational experiments which show that this modification consistently reduces the time needed to solve the model to optimality by 30-50\% in larger problem instances.

Additionally, this constraint can also be added to Model 1. Aside from greatly reducing the number of variables used, the original constraints of Model 1 can be changed from fort neighborhood cover constraints to cover constraints over merely the junctions contained by a fort neighborhood. In large problem instances or graphs with relatively long junction paths (often seen in electrical networks), this can reduce the amount of memory needed to hold the constraint set, further improving runtime performance.

\subsection{Computational Experiments}
The graphs used in these experiments consist of networks from standard IEEE test cases\footnote{Publicly available at the Power Systems Test Case Achieve: \url{https://labs.ece.uw.edu/pstca/}.}, several large networks created by the Pan European Grid Advanced Simulation and State Estimation (PEGASE) project \cite{pegase} meant to simulate the scale and complexity of the European high voltage transmission network, and several publically available power grid data sets including the US Western Interconnection network compiled by Watts and Strogatz~\cite{watts-strogatz}. All of the graphs considered are available for download at the LIINES Smart Grid Test Case Repository~\footnote{The LIINES Smart Grid Test Case Repository and all graphs used can be found at \url{http://amfarid.scripts.mit.edu/Datasets/SPG-Data/index.php.}}.

In order to ensure fair comparisons of the methods, we have implemented or reimplemented all algorithms on the same software platforms and have run all computational experiments using the same hardware. The experiments presented here were implemented using Python 3.5 and Gurobi 7.5.2~\cite{gurobi}, and were run on a laptop with 16GB of RAM and an i7 2.80GHz processor. The constraints of Model 1 were generated using the Gurobi callbacks. 

In Figure \ref{tab:sep-results}, we compare the performance of the power domination set cover model (Model 1) when constraints are generated with Model 2 and Model 3. In all test cases, Model 1 reached optimality more quickly and with fewer added constraints when Model 3 was used. In all experiments the constraint set of Model 1 was initialized via the process described in Section \ref{sec:M-init} and Constraint~(\ref{cons:8}) was included in the model.

In Figure \ref{tab:full-results}, we compare the runtime performance of the infection IP model presented by Brimkov et al. \cite{conn-pd}, the modified infection model discussed in Section \ref{sec:inf-model} (denoted in the table as Inf.*), and Model 1. As before, the constraint set of Model 1 was initialized via the process described in Section \ref{sec:M-init} and Constraint~(\ref{cons:8}) was included. Violated constraints of Model 1 were separated using Model 3.

\renewcommand{\arraystretch}{1.2}
\begin{figure}[H]
  \centering
  \footnotesize
  \[ \begin{tabular}{|c|c|c|c|c|c|c|c|c|c|}
		\hline
 		\multicolumn{10}{|c|}{Comparison of Set Cover Constraint Generation Methods} \\
    \hline
    \multicolumn{6}{|c|}{\multirow{2}{*}{Graph Test Instances}} & \multicolumn{2}{c|}{Set Cover Model} & \multicolumn{2}{c|}{Set Cover Model} \\
    \multicolumn{6}{|c|}{}& \multicolumn{2}{c|}{(using Model 2)} & \multicolumn{2}{c|}{(using Model 3)} \\
 		\hline
 	  Name         & $n$  & $m$   & $J$ & $\mathscr{M}$ &$\gamma_P$& Time & Sep. & Time   & Sep. \\
	  \hline 
	  IEEE Bus 14	 & 14   & 20    & 7   & 0 & 2        & 0.042    & 3       &\textbf{0.007} & 2   \\ 
	  IEEE Bus 30  & 30   & 41    & 12  & 1 & 3        & 0.019		 & 4       &\textbf{0.012} & 2   \\ 
	  IEEE Bus 57	 & 57   & 78    & 24  & 1 & 3        & 0.034		 & 6       &\textbf{0.029} & 5   \\ 
	  IEEE Bus 118 & 118  & 179   & 55  & 2 & 8        & 0.19	   & 26      &\textbf{0.12}  & 10  \\ 
	  IEEE Bus 300 & 300  & 409   & 155 & 14 & 30       & 1.09     & 63      &\textbf{0.65}  & 30  \\ 
    PEGASE 1354  & 1354 & 1710  & 496 & 140 & 176      & --       & 4785    &\textbf{1.55}  & 49  \\ 
    Polish 2383  & 2383 & 2886  & 776 & 95 & 203      & --       & 14083   &\textbf{52.73} & 14  \\ 

    \hline
  \end{tabular} \]
  \caption{\footnotesize Runtime performance comparison of the set cover model (Model 1) with constraints generated by Model 2 and Model 3. In both cases, Model 1 is initialized with the process discussed in Section \ref{sec:M-init}. A timeout period of 1 hour was used.}
  \label{tab:sep-results}
\end{figure}
\vspace{-.5cm}
\begin{figure}[H]
  \centering
  \footnotesize
  \[ \begin{tabular}{|c|c|c|c|c|c|c|c|c|c|}
		\hline
 		\multicolumn{10}{|c|}{Comparison with Previous Power Domination Computational Methods} \\
    \hline
 		\multicolumn{6}{|c|}{Graph Test Instances} & \multicolumn{1}{c|}{Inf. \cite{conn-pd}} & \multicolumn{1}{c|}{Inf.*} & \multicolumn{2}{c|}{Set Cover Model} \\
 		\hline
 	  Name         & $n$  & $m$   & $J$  &$\mathscr{M}$&$\gamma_P$& Time & Time & Time & Sep. \\
	  \hline 
	  IEEE Bus 14	 & 14   & 20    & 7    & 0   &  2    & 0.10     & 0.06     & \textbf{0.007}   & 2   \\ 
	  IEEE Bus 30  & 30   & 41    & 12   & 1   &  3    & 0.08		 & 0.08     & \textbf{0.012}	  & 2   \\ 
	  IEEE Bus 57	 & 57   & 78    & 24   & 1   &  3    & 0.34		 & 0.34     & \textbf{0.029}	  & 5   \\ 
	  IEEE Bus 118 & 118  & 179   & 55   & 2   &  8    & 3.97	   & 2.71     & \textbf{0.12}	  & 10  \\ 
	  IEEE Bus 300 & 300  & 409   & 155  & 14  &  30   & 64.64    & 20.74    & \textbf{0.65}	  & 30  \\ 
    PEGASE 1354  & 1354 & 1710  & 496  & 140 &  176  & 10.87    & 4.82     & \textbf{1.55}    & 49  \\ 
    Polish 2383  & 2383 & 2886  & 776  & 95  &  203  & 804.31   & 482.52   & \textbf{52.73}   & 14  \\ 
    \multirow{2}{*}{\shortstack{US Western\\Interconnection}}
                 &\multirow{2}{*}{4941}& \multirow{2}{*}{6594}& \multirow{2}{*}{2059}& \multirow{2}{*}{321} & \multirow{2}{*}{494} & \multirow{2}{*}{--}& \multirow{2}{*}{--} &\multirow{2}{*}{\textbf{185.00}} &\multirow{2}{*}{31} \\
    & & & & & & & & &     \\ 
    PEGASE 9241  & 9241 & 14207 & 3800 & 527 & 811 & -- & -- & \textbf{995.54}   & 299 \\ \hline

  \end{tabular} \]
  \caption{\footnotesize Runtime performance comparison of the infection model proposed by Brimkov et al. \cite{conn-pd} (Model 4), the modified infection model, and the set cover model with constraints generated using Model 3. A timeout period of 6 hours was used.}
  \label{tab:full-results}
\end{figure}

{\let\thefootnote\relax\footnote{{$J$ is used to denote $|J(G)|$, $\mathscr{M}$ is used to denote the number of constraints initially included in Model 1, and the Sep. column provides the number of constraints which were separated in each test case. All tests are reported in seconds, the best time for each test case is in bold.}}}

\renewcommand{\arraystretch}{1}

In the larger test cases, the majority of Model 1's constraints are added during the initialization process. When Constraint (\ref{cons:8}) is added to Model 1, the initial constraints introduce either required vertices (fort neighborhoods of type I and II) or cover constraints over two vertices (fort neighborhoods of type III). As a result, the initially added constraints can significantly reduce the dimension of the Model 1 feasible region.

\section{Conclusions}
In this paper, we study the power dominating set problem from the perspective of zero forcing forts and their closed neighborhoods. Our main contribution is a novel computational method for computing the power domination number of general graphs and identifying minimum cardinality power dominating sets. Our method is an IP formulation which, unlike the existing computational methods for power domination, can be solved via a constraint generation strategy. To define this method and verify its correctness, we provide an equivalent description of power dominating sets in the context of fort neighborhoods, introduce a useful vertex partitioning scheme, and prove several structure properties of fort neighborhoods. We also provide a discussion of the computational complexity of identifying minimum weight fort neighborhoods and show that the minimum weight fort neighborhood decision problem is NP-Complete if even one vertex is known to be in the fort neighborhood.

Additionally, we give provide computational experiments demonstrating the significantly improved runtime performance of the proposed method. In these experiments, the proposed algorithm was able to determine the power dominating number of both the 4941 node US Western Interconnection graph and the 9241 node PEGASE 9241 test case, which were both previously unknown, to the best knowledge of the authors, and are an order of magnitude larger than the test cases previously seen in power domination literature. Moreover, the techniques presented in this paper can be easily extended to many other dynamic graph processes. Similar to power domination and zero forcing phenomena, graph coloring rules have emerged in other applications including $k$-thresholding problems which model disease and opinion transfer~\cite{thresholding1, dreyer} and target set selection problems arising in optimal advertising resource allocation~\cite{tss2, tss1}. For these problems, structures similar to forts or fort neighborhoods which impede propagation could be identified and IPs similar to Model 1 could be formulated. 

We conclude this section by discussing several possible directions of future work. First, it may be possible to improve the overall performance of our method by generating Model 1's constraints with a polynomial time complexity heuristic. While the power dominating set problem is known to be NP-Complete, violated constraints are identified through the use of an auxiliary IP. In many constraint generation algorithms, effective violated constraints can be quickly generated through the use of heuristic algorithms, benefiting runtime performance. Second, the initial set of fort neighborhood constraints used in the set cover model was generated through the identification of fort neighborhoods with special structures. These special types of fort neighborhoods chosen as they could be quickly identified and seem to occur regularly in electrical network graphs. Other regularly occurring structures could be similarly exploited. In addition, power domination has also been recently studied in the context of special graphs arising in chemistry \cite{pdchem2, pdchem1}. In new problem domains it is likely that other special structures will be relevant and different initialization techniques will be needed.

Another direction of potential future work is the determination of the computational complexity of the subproblems solved by Model 2 and 3. While we believe that these subproblems are indeed NP-Hard to solve in general graphs, it has not yet been shown. A recent complexity result posed by Shitov~\cite{shitov} implies that the minimum cardinality fort problem is NP-Hard; it may be possible to extend these results to fort neighborhoods. Finally, Model 3 generates violated fort neighborhood constraints with minimum cardinality vertex sets. This ensures that the constraints Model 3 generates satisfy half of the necessary and sufficient conditions for cover constraints to be facets, as proposed by Balas and Ng~\cite{balas-ng}. It is likely possible that Model 3 can be extended to generate facet inducing violated inequalities for the master problem, and a similar approach was taken by Fast~\cite{fast-thesis} for the zero forcing set problem. \newline

\section*{Acknowledgments}
This work was supported by the National Science Foundation, under Grant DMS-1720225.

\bibliographystyle{plain}
\bibliography{main.bib}

\end{document}